\documentclass[10pt]{article}

\usepackage{amsfonts,graphicx,amsmath,amssymb,amsthm,geometry,bbm,hyperref,nicefrac,enumerate,comment,fontawesome,mathtools}

\usepackage{xcolor}
\DeclareMathOperator*{\argmax}{arg\,max}
\DeclareMathOperator*{\argmin}{arg\,min}

\newcommand{\I}{\mathbb{I}}
\newcommand{\R}{\mathbb{R}}
\newcommand{\N}{\mathbb{N}}
\newcommand{\E}{\mathbb{E}}

\renewcommand{\P}{\mathbb{P}}
\newcommand{\Q}{\mathbb{Q}}
\newcommand{\A}{\mathbb{A}}
\newcommand{\T}{\mathcal{T}}
\newcommand{\Tn}{\mathcal{T}_n}
\newcommand{\F}{\mathcal{F}}
\newcommand{\Fn}{\mathcal{F}_n}
\newcommand{\e}{\text{e}}
\renewcommand{\d}{\text{d}}
\newcommand{\D}{\mathcal{D}}
\usepackage{array}
\newcolumntype{P}[1]{>{\centering\arraybackslash}p{#1}}

\usepackage{geometry}
\newtheorem{remark}{Remark}[section]

\newtheorem{theorem}{Theorem}[section]

\newtheorem{example}{Example}[section]

\theoremstyle{definition}

\newcommand{\footremember}[2]{%
    \footnote{#2}
    \newcounter{#1}
    \setcounter{#1}{\value{footnote}}%
}
\newcommand{\footrecall}[1]{%
    \footnotemark[\value{#1}]%
} 

\begin{document}

\title{A deep learning approach for computations of exposure profiles for high-dimensional Bermudan options}

\author{%
  Kristoffer Andersson\footremember{alley}{Research Group of Scientific Computing, Centrum Wiskunde \& Informatica} 
  \and Cornelis W. Oosterlee\footrecall{alley} \footremember{trailer}{Delft Institute of Applied Mathematics (DIAM), Delft University of Technology}
 }

\maketitle
\begin{abstract}
In this paper, we propose a neural network-based method for approximating expected exposures and potential future exposures of Bermudan options. In a first phase, the method relies on the Deep Optimal Stopping algorithm (DOS) proposed in \cite{DOS}, which learns the optimal stopping rule from Monte-Carlo samples of the underlying risk factors. Cashflow paths are then created by applying the learned stopping strategy on a new set of realizations of the risk factors. Furthermore, in a second phase the cashflow-paths are projected onto the risk factors to obtain approximations of pathwise option values. The regression step is carried out by ordinary least squares as well as neural networks, and it is shown that the latter produces more accurate approximations.

The expected exposure is formulated, both in terms of the cashflow-paths and in terms of the pathwise option values and it is shown that a simple Monte-Carlo average yields accurate approximations in both cases. The potential future exposure is estimated by the empirical $\alpha$-percentile.

Finally, it is shown that the expected exposures, as well as the potential future exposures can be computed under either, the risk neutral measure, or the real world measure, without having to re-train the neural networks.

\end{abstract}
\tableofcontents

\section{Introduction}
The exposure of a financial contract is the maximum amount that an investor stands to lose if the counterparty is unable to fulfill its obligations, for instance, due to a default. This means that, in addition to the market risk, a so-called counterparty credit risk (CCR) needs to be accounted for. Furthermore, the liquidity risk, which is the risk arising from potential costs of unwinding a position, is also closely related to the financial exposure. Over the counter (OTC) derivatives, \textit{i.e.,} contracts written directly between counterparties, instead of through a central clearing party (CCP), are today mainly subject to so-called valuation adjustments (XVA\footnote{ X represents arbitrary letters, \textit{e.g.,} ''C'' for credit valuation adjustment, ''F'' for funding valuation adjustment, etc.}). These valuation adjustments aim to adjust the value of an OTC derivative for certain risk factors, \textit{e.g.,} credit valuation adjustment (CVA), adjusting the value for CCR, funding valuation adjustment (FVA), adjusting for funding cost of an uncollateralized derivative or capital valuation adjustment (KVA), adjusting for future capital costs. The financial exposure is central in calculations of many of the XVAs (for an in-depth overview of XVAs, we refer to \cite{XVA_gregory} and \cite{XVA_green}). In this paper, we therefore focus on computations of financial exposures of options.

For a European style option, the exposure is simply the option value at some future time, given all the financial information available today.  If $t_0$ is the time today and $X_t$ the $d$-dimensional vector of underlying risk factors of an option, we define the exposure of the option, at time $t>t_0$, as the random variable\footnote{ Usually, the exposure is defined as $\max\{V_t(X_t),0\}$, but since we only consider options with non-negative values, the max operator is omitted.} \begin{equation*}
    \text{E}^\text{Eur}_t= V_{t}(X_t).
\end{equation*}
However, for derivatives with early-exercise features, we also need to take into account the possibility that the option has been exercised prior to $t$, sending the exposure to zero. The most well-known of such options is arguably the American option, which gives the holder (or buyer) the right to exercise the option at any time between $t_0$ and maturity $T$. In this paper, the focus is on the Bermudan option which gives the holder the right to exercise at finitely many, predefined exercise dates between $t_0$ and maturity $T$. For early-exercise options, the exposure needs to be adjusted for the possibility that the stopping time $\tau$, representing the instance of time at which the option is exercised, is smaller than or equal to $t$. This leads to the following definition of the exposure 
\begin{equation*}
   \text{E}^{\text{Ber}}_t= V_{t}(X_t)\mathbbm{I}_{\{\tau> t\}},
\end{equation*}
where $\I_{\{\cdot\}}$ is the indicator function. Two of the most common ways to measure the exposure are the expected exposure (EE), which is the expected future value of the exposure, and the potential future exposure (PFE), which is some $\alpha-$percentile of the future exposure (usually $\alpha=97.5\%$ for the upper, and $\alpha=2.5\%$ for the lower tails of the distribution). To accurately compute EE and PFE of a Bermudan option, it is crucial to use an algorithm which is able to accurately approximate, not only $V_t(X_t)$, but also $\I_{\{\tau> t\}}$. 

It is common to compute the EE and the PFE with simulation-based methods, usually from realizations of approximate exposures, from which the distribution of $\text{E}_t$ is approximated by its empirical counterpart. A generic scheme for approximation of exposure profiles is given below:
\begin{enumerate}
    \item At each exercise date, find a representation, $v_t:\R^d\to\R$, of the true value function $V_t(\cdot)$;
    \item Generate trajectories of the underlying risk factors, and at each exercise date, evaluate\footnote{We define a filtered probability space in the next section from which $\omega$ is an element. In this section, $\omega$ should be viewed as an outcome of a random experiment.} $v_t(X_t(\omega))$;
    \item At the earliest exercise date where $v_t(X_t(\omega))\leq g(X_t(\omega))$, where $g$ is
    the pay-off function, set $\tau(\omega)=t$;
    \item The distribution of $\text{E}_t^{\text{Ber}}$, is then approximated by the empirical distribution created by many trajectories of the form $v_t(X_t(\omega))\I_{\{\tau(\omega)>t\}}$. For instance, if the target statistic is the EE, the estimation is the sample mean of the exposure paths.
\end{enumerate}
Step 1 above corresponds to the option valuation problem, which can be tackled by several different methods, all with their own advantages and disadvantages. For instance, the value function can be approximated by solving an associated partial differential equation (PDE), which is done in \textit{e.g.,} \cite{Forsyth},\cite{Reisinger},\cite{Vasquez},\cite{Hout} and \cite{Hout_book}, or the value function can be approximated by a Fourier transform methodology, which is done in \text{e.g.,} \cite{Fang}, \cite{Zy} and \cite{Fang2}. Furthermore, classical tree-based methods such as \cite{tree_broadie}, \cite{Rubinstein} and \cite{Jack}, can be used. These types of methods are, in general, highly accurate but they suffer severely from the curse of dimensionality, meaning that they are computationally feasible only in low dimensions (say up to $d=4$), see \cite{COD}. In higher dimensions, Monte-Carlo-based methods are often used, see \textit{e.g.,} \cite{Binomial_price},\cite{LSM},\cite{BG},\cite{maxcall} and \cite{SGBM}. Monte-Carlo-based methods can generate highly accurate option values at $t_0$, \textit{i.e.,} $v_{t_0}(X_{t_0}=x_{t_0})$, given that all trajectories of the underlying risk factors are initiated at some deterministic $x_{t_0}$. However, at intermediate exercise dates, $v_t(X_t(\omega))$ might not be an equally good approximation, due to the need of cross sectional regression. We show with numerical examples that, even though the approximation is good on average, the option value is underestimated in certain regions, which is compensated by overestimated values in other regions. For European options, this seems to have minor effect on EE and PFE, but for Bermudan options the effect can be significant due to step 3 above. To provide an intuitive understanding of the problem, we give an illustrative example. Assume that $v_t$ underestimates the option value in some region A, which is close to the exercise boundary, and overestimates the option value in some region B, where it is clearly optimal to exercise the option. The effect on the exposure would be an overestimation in region A, since underestimated option values would lead to fewer exercised options. In region B, the exposure would be zero in both cases since all options in that region would be exercised immediately. In total, this would lead to overestimated exposure. In numerical examples this is, of course, more involved and we see typically several regions with different levels of under/overestimated values. This makes the phenomenon difficult to analyze since the effect may lead to underestimated exposure, unchanged exposure (from offsetting effects), or overestimated exposure. In addition to the classical regression methods, with \textit{e.g.,} polynomial basis functions, which are cited above, there are several papers in which a neural network plays the role of the basis functions, see \textit{e.g.,} \cite{Kohler},\cite{Cheridito}, \cite{Laypere} and \cite{SGBMNN}.

In this paper we suggest the use of the Deep Optimal Stopping (DOS) algorithm, proposed in \cite{DOS}, to approximate the optimal stopping strategy. The DOS algorithm approximates the optimal stopping strategy by expressing the stopping time in terms of binary decision functions, which are approximated by deep neural networks. The DOS algorithm is very suitable for exposure computations, since the stopping strategy is computed directly, not as a consequence of the approximate value function, as is the case with most other algorithms. This procedure leads to highly accurate approximations of the exercise boundaries. Furthermore, we propose a neural network-based regression method (NN-regression) to approximate $V_t(\cdot)$, which relies on the stopping strategy, approximated by the DOS algorithm. Although NN-regression is needed in order to fully approximate the distributions of future exposures, it turns out that for some target statistics, it is in fact sufficient to approximate the optimal stopping strategy. This is, for instance, the case for some estimates of the EE. 

Another advantage of the method is that the EE and the PFE easily can be computed under the risk neutral measure, as well as the real world measure. When the EE is used to compute CVA, the calculations should be done under the risk neutral measure, since CVA is a tradable asset and any other measure\footnote{ For fixed numéraire.} would create opportunities of arbitrage. For other applications, such as for computations of KVA (valuation adjustment to account for future capital costs associated to a derivative), the EE should be calculated under the real world measure\footnote{ In this setting, the EE is an expectation under the real world measure, which is conditioned on a filtration generated by the underlying risk factors.}. The reason for this is that the KVA, among other things, depends on the required CCR capital which is a function of the EE, which in this case, ideally, should be calculated under the real world measure. For an explanation of KVA in general, see \textit{e.g.,} \cite{KVA} and for a discussion about the effect of different measures in KVA computations see \cite{PsandQs}.

 Finally, we emphasize that, even though the focus of this paper is to compute exposure profiles, the algorithms proposed are flexible and with small adjustments other kind of risk measures could be computed. The first reason for only studying exposure profiles is, as mentioned above, that it is an important building block in computations (or approximations) of many XVAs as well as expected shortfall. Another reason is that there are a number of similar studies in the existing literature (see \textit{e.g.,} \cite{SGBM_Heston}, \cite{Yanbin}, \cite{SGBM_P}). However, it should be pointed out that the algorithms proposed in this paper are not limited to computations of exposures. As an example, consider the CVA of a Bermudan option, expressed as the difference between the risky (including CCR) and the risk-free (excluding CCR) option values. To accurately compute the risky option value, one needs to adjust the exercise strategy for the risk of a default of the counterparty. Therefore, it is not sufficient to compute the exposure of the option when computing the CVA (for a discussion see \textit{e.g.,} \cite{CVA_BER1}, \cite{CVA_BER2}). With minor adjustments to the proposed algorithm, one could compute this kind of advanced CVA (as opposed to the approximation, which is based on the exposure). 

This paper is structured as follows: In Section \ref{sec2}, the mathematical formulation of a Bermudan option and its exposure profiles are given. Furthermore, the Bermudan option, as well as the exposure profiles are reformulated
to fit the algorithms introduced in later sections. The DOS algorithm is described in Section \ref{learnSD}, and we propose some adjustments to make the training procedure more efficient. In Section \ref{learn_pw_optionvalues}, we present a classical ordinary least squares regression-based method (OLS-regression), as well as the NN-regression to approximate pathwise option values. In Section \ref{EEapprox}, the EE and PFE formulas are reformulated in a way such that they can easily be estimated by a combination of the algorithms presented in
Sections \ref{learnSD} and \ref{learn_pw_optionvalues}, and a simple Monte-Carlo sampling. Finally, in Section \ref{num_res}, numerical results of the different
algorithms, are presented for Bermudan options following the Black--Scholes dynamics as well as the Heston model dynamics.

\section{Problem formulation}\label{sec2}
For $d,N\in\N\setminus\{0\}$, let $X=(X_{t_n})_{n=0}^N$ be an $\R^d-$valued discrete-time Markov process on a complete probability space $(\Omega,\,\mathcal{F},\,\A)$. The outcome set $\Omega$ is the set of all possible realizations of the stochastic economy, $\F$ is a $\sigma-$algebra on $\Omega$ and we define $\Fn$ as the sub-$\sigma$-algebra generated by $(X_{t_m})_{m=0}^n$. With little loss of generality, we restrict ourselves to the case when $X$ is constructed from time snap shots of a continuous-time Markov process at monitoring dates $\{t_0,\,t_1,\,\ldots,\,t_N\}$. The probability measure $\A$ is a generic notation, representing either the real world measure, or the risk neutral measure, denoted $\P$ and $\Q$, respectively. 

If not specifically stated otherwise, equalities and inequalities of random variables should be interpreted in an $\omega$-wise sense.
\subsection{Bermudan options, stopping decisions and exercise regions}\label{Beropt_stopdec_exreg}
A Bermudan option is an exotic derivative that gives the holder the opportunity to exercise the option at a finite number of exercise dates, typically one per month. We define the exercise dates as $\mathbbm{T}=\{t_0,\,t_1,\,\ldots,\,t_N\}$, which for simplicity coincide with the monitoring dates. Furthermore, for $t_n\in\mathbbm{T}$, we let the remaining exercise dates be defined as $\mathbbm{T}_n=\{t_n,\,t_{n+1}\,\ldots,\,t_N\}$. Let $\tau$ be an $X-$stopping time, \textit{i.e.}, a random variable defined on $(\Omega,\,\mathcal{F},\,\A)$, taking on values in $\mathbbm{T}$ such that for all $t_n\in\mathbbm{T}$, it holds that the event $\{\tau = t_n\}\in\Fn$. Assume a risk-free rate $r\in\R$ and let the risk-less savings account process, $M(t)=\e^{r(t-t_0)}$, be our numéraire.
For all $t\in\mathbbm{T}$, let\footnote{ Allowing for different pay-off functions at different exercise dates makes the framework more flexible. We can \textit{e.g.}, let $g_t$ be the pay-off function for an entire netting set of derivatives with different maturities.} $g_t:\R^d\to\R$ be a measurable function which returns the immediate pay-off of the option, if exercised at market state $(t,X_t=x\in\R^d)$. The initial value of a Bermudan option, \textit{i.e.,} the value at market state $(t_0,\,X_{t_0}=x_{t_0}\in\R^d)$, is given by
\begin{equation}
    \label{V0}
    V_{t_0}(x_{t_0}) = \sup_{\tau\in\T}\E_\Q\left[\e^{-r(\tau-t_0)}g_\tau(X_\tau)\,|\,X_{t_0}=x_{t_0}\right],
\end{equation}
where $\T$ is the set of all $X-$stopping times. Assuming the option has not been exercised prior to some $t_n\in\mathbbm{T}$, the option value at $t_n$ is given by 
\begin{equation}
    \label{Vn}
    V_{t_n}(X_{t_n}) = \sup_{\tau\in\Tn}\E_\Q^n\left[\e^{-r(\tau-t_n)}g_\tau(X_\tau)\right],
\end{equation}
where $\Tn$ is the set of all $X-$stopping times greater than or equal to $t_n$ and we define $\E_\A^n[\,\cdot\,]=\E_\A[\,\cdot\,|\Fn]$. To guarantee that \eqref{V0} and \eqref{Vn} are well-posed, we assume that for all $t\in\mathbbm{T}$ it holds that
\begin{equation*}
\label{integrability_g}
    \E_\Q\left[|g_t(X_t)|\right]<\infty.
\end{equation*}

To concretize \eqref{V0} and \eqref{Vn}, we view the problem from the holder's perspective. At each exercise date, $t_n\in\mathbbm{T}$, the holder of the option is facing the decision whether or not to exercise the option, and receive the immediate pay-off. Due to the Markovian nature of $X$, the decisions are of Markov-type (Markov decision process (MDP)), meaning that all the information needed to make an optimal decision\footnote{In the sense of using a decision policy such that the supremum in \eqref{V0} or \eqref{Vn} is attained.} is contained in the current market state, $(t_n,\,X_{t_n})$. With this in mind, we define for each $t_n\in\mathbbm{T}$, the exercise region, $\mathcal{E}_n$, in which it is optimal to exercise, and the continuation region, $\mathcal{C}_n$, in which it is optimal to hold on, by
\begin{equation*}
    \mathcal{E}_n = \left\{x\in\mathbbm{R}^d\,|\,V_{t_n}(x)=g_{t_n}(x)\right\},\quad\mathcal{C}_n = \left\{x\in\mathbbm{R}^d\,|\,V_{t_n}(x)>g_{t_n}(x)\right\}.
\end{equation*}
Note that $\mathcal{E}_n\cup\mathcal{C}_n=\R^d$ and $\mathcal{E}_n\cap\mathcal{C}_n=\emptyset$. 

Below, we give a short motivation for how these decisions can be expressed mathematically and how we can formulate a stopping time in terms of (stopping) decisions at each exercise date. For a more detailed motivation we refer to \cite{DOS}. We introduce the following notation for measurable functions\footnote{We assume measurable spaces $(A_1,\mathcal{A}_1)$ and $(A_2,\mathcal{A}_2)$ and measurable functions with respect to $\sigma$-algebras $\mathcal{A}_1$ and $\mathcal{A}_2$. This assumption holds in all cases in this paper.}
\begin{equation*}
    \D(A_1;A_2)=\{\,f\colon A_1\to A_2\,|\,f\ \text{measurable}\}. 
\end{equation*}
Furthermore, for $P\in\N$, let $\D(A;B)^P$ denote the $P:$th Cartesian product of the set $\D(A_1;A_2)$. Define the decision functions $f_0,\,f_1,\,\ldots,\,f_{N}\in \D(\R^d;\,{\{0,1\}})$, with $f_N\equiv1$, and denote $\boldsymbol{f}_n=(f_n,\,f_{n+1},\,\ldots,\,f_N)$ with $\boldsymbol{f}=\boldsymbol{f}_0$. An $X-$stopping time can then be defined as
\begin{equation}\label{tauk}
    \tau_n(\boldsymbol{f}_n) = \sum_{m=n}^Nt_mf_m(X_{t_m})\prod_{j=n}^{m-1}(1-f_j(X_{t_j})),
\end{equation}
where the empty product is defined as 1. We emphasize that $\tau_n(\boldsymbol{f}_n)=\tau_n\left(f_n(X_{t_n}),\,f_{n+1}(X_{t_{n+1}})\ldots,\,f_N(X_{t_N})\right)$ but to make the notation cleaner, we do not specify this dependency explicitly. When it is not clear from the context which process the decision function $\boldsymbol{f}_n$ is acting on, we will denote the stopping time by $\tau_n\left(\boldsymbol{f}_n(X)\right)$, where we recall that $X=(X_{t_m})_{m=n}^N$. As a candidate for optimal stopping decisions, we define
 \begin{equation}\label{opt_dec}
    \boldsymbol{f}^* \in \argmax_{\boldsymbol{f}\in\D(\R^d;\,{\{0,1\}})^{N+1}}\E_\Q\left[\e^{-r(\tau(\boldsymbol{f})-t_0)}g_{\tau(\boldsymbol{f})}\left(X_{\tau(\boldsymbol{f})}\right)\,|\,X_{t_0}=x_0\right].
\end{equation}
We define the fair price of an option that follows the strategy constructed by combining \eqref{tauk} and \eqref{opt_dec} as
\begin{equation}\label{Vtaustar}
    V^*_{t_n}(X_{t_n})=\E_\Q^n\left[\e^{-r(\tau_n(\boldsymbol{f}_n^*)-t_n)}g_{\tau_n(\boldsymbol{f}_n^*)}\left(X_{\tau_n(\boldsymbol{f}_n^*)}\right)\right].
\end{equation}
The following theorem states that \eqref{Vtaustar}, in fact, coincides with the fair price of the option as defined in \eqref{Vn}.

\begin{theorem}\label{VstarV}
For all $t_n\in\mathbbm{T}$, and $V$ and $V^*$ as defined in \eqref{Vn} and \eqref{Vtaustar}, respectively, it holds that\footnote{ We recall that equalities and inequalities are in an $\mathbbm{A}-$almost sure sense.}
\begin{equation*}
    V_{t_n}(X_{t_n})=V^*_{t_n}(X_{t_n}).
\end{equation*}
\begin{proof}
Note that $V_{t_N}(X_{t_N})=V^*_{t_N}(X_{t_N})=g_{t_N}(X_{t_N})$. The proof is carried out by induction and we assume that for some $t_{n+1}\in\mathbbm{T}$ it holds that\begin{equation}\label{assumption_1_thm1}V_{t_{n+1}}(X_{t_{n+1}})=V^*_{t_{n+1}}(X_{t_{n+1}}).\end{equation}
We can rewrite $V^*_{t_n}(X_{t_n})$ as
\begin{align}
    V^*_{t_n}(X_{t_n})=&\E_\Q^n\left[\e^{-r(\tau_n(\boldsymbol{f}_n^*)-t_n)}g_{\tau_n(\boldsymbol{f}_n^*)}(X_{\tau_n(\boldsymbol{f}_n^*)})\right]\nonumber\\
    =&g_{t_n}(X_{t_n})f_n^*(X_{t_n})\label{thm_1_E2}+(1-f_n^*(X_{t_n}))\e^{-r(t_{n+1}-t_n)}\nonumber\\
    &\times\E_\Q^n\left[\e^{-r(\tau_{n+1}(\boldsymbol{f}_{n+1}^*)-t_{n+1})}g_{\tau_{n+1}(\boldsymbol{f}_{n+1}^*)}(X_{\tau_{n+1}(\boldsymbol{f}_{n+1}^*)})\right].
\end{align}
By the law of total expectation and the assumption \eqref{assumption_1_thm1}, the last conditional expectation satisfies
\begin{equation}\label{thm_1_E}
    \E_\Q^n\left[\e^{-r(\tau_{n+1}(\boldsymbol{f}_{n+1}^*)-t_{n+1})}g_{\tau_{n+1}(\boldsymbol{f}_{n+1}^*)}(X_{\tau_{n+1}(\boldsymbol{f}_{n+1}^*)})\right]=\E_\Q^n\left[V^*_{t_{n+1}}(X_{t_{n+1}})\right]=\E_\Q^n\left[V_{t_{n+1}}(X_{t_{n+1}})\right].
\end{equation}
We insert the rightmost part of \eqref{thm_1_E} in \eqref{thm_1_E2} and note that $\mathbbm{I}_{\{\,\cdot\,\in\mathcal{E}_n\}}\in\D(\R^d;\,{\{0,1\}})$, which implies that
\begin{align*}
V^*_{t_n}(X_{t_n})\geq& \mathbbm{I}_{\{X_{t_n}\in\mathcal{E}_n\}} g_{t_n}(X_{t_n})+\mathbbm{I}_{\{X_{t_n}\in\mathcal{C}_n\}}\e^{-r(t_{n+1}-t_n)}\E_\Q^n\left[V_{t_{n+1}}(X_{t_{n+1}})\right]\\
=&V_{t_n}(X_{t_n}).
\end{align*}
Moreover, $V^*_{t_n}(X_{t_n})\leq V_{t_n}(X_{t_n})$ and therefore we conclude that $V^*_{t_n}(X_{t_n})= V_{t_n}(X_{t_n})$.
\newline\newline
OR:
\newline\newline
Follows directly from \cite[Theorem 1]{DOS}.
\end{proof}
\end{theorem}

\subsection{Exposure profiles}\label{EP_sec}

For $t_n\in\mathbbm{T}$ and $\alpha\in(0,\,1)$, we define the expected exposure (EE) and the potential future exposure (PFE) under the generic probability measure $\A$ as
\begin{align}\label{EE}
        \text{EE}_\A(t_n) &= \E_\A\left[V_{t_n}(X_{t_n})\I_{\{\tau(\boldsymbol{f}^*)> t_n\}}\,|\,X_{t_0}=x_{t_0}\right],\\\label{PFE}
    \text{PFE}_\A^\alpha(t_n) &= \inf\left\{ a\in\R \,|\,\A\left(V_{t_n}(X_{t_n})\I_{\{\tau(\boldsymbol{f}^*)> t_n\}}\leq a\,|\,X_{t_0}=x_{t_0}\right)\geq\alpha\right\}.
\end{align} 
Note that the option value, given by \eqref{Vn}, is a conditional expectation of future cashflows, which is by definition measured with $\Q$. The exposure profiles under the $\P-$measure, on the other hand, are statistics of the option value at some future time, $t_n$, under the assumption that the conditional distribution $X_{t_m}$ conditional on $X_{t_0}=x_{t_0}$ is considered under the $\P-$measure.

In the theorem below, the expected exposures are reformulated in terms of discounted cashflows and the decision functions introduced in Subsection \ref{Beropt_stopdec_exreg}. It will become clear in later sections that this is a more tractable form for the algorithms used in this paper.
\begin{theorem}\label{EEtheorem}
Let $\Q$ and $\P$ be probability measures on the measurable space $(\Omega,\,\F)$ and assume that the laws of X under $\Q$ and $\P$ are absolutely continuous with respect to the Lebesgue measure. Then, the expected exposure, \eqref{EE} under the $\Q-$ and $\P-$measures, satisfies
\begin{align}\label{EEQ}
\emph{EE}_\Q(t_n)=&\E_\Q\left[\e^{-r(\tau_n(\boldsymbol{f}_n^*)-t_n)}g_{\tau_n(\boldsymbol{f}_n^*)}(X_{\tau_n(\boldsymbol{f}_n^*)})\I_{\{\tau(\boldsymbol{f}^*)> t_n\}}\,|\,X_{t_0}=x_{t_0}\right],\\\label{EEP}
\emph{EE}_\P(t_n)=&\E_\Q\left[\e^{-r(\tau_n(\boldsymbol{f}_n^*)-t_n)}g_{\tau_n(\boldsymbol{f}_n^*)}(X_{\tau_n(\boldsymbol{f}_n^*)})\I_{\{\tau(\boldsymbol{f}^*)> t_n\}}l(X_{t_n},\,X_{t_{n-1}},\,\ldots,\,X_{t_0})\,|\,X_{t_0}=x_{t_0}\right],
\end{align}
where $l(y_{n},\,y_{n-1},\,\ldots,\,y_{0})=\prod_{k=1}^{n}\frac{p_{X_{t_k}|X_{t_{k-1}}}(y_{k}|y_{k-1})}{{q_{X_{t_k}|X_{t_{k-1}}}(y_k|y_{k-1})}}$, with $q_{X_{t_k}|X_{t_{k-1}}}(y_k|y_{k-1})$ and\\ $p_{X_{t_k}|X_{t_{k-1}}}(y_k|y_{k-1})$ being transition densities for $X$ under the measures $\Q$ and $\P$, respectively. Note that $l(y_{n},\,y_{n-1},\,\ldots,\,y_{0})$ is the Radon--Nikodym derivative of $\P$ with respect to $\Q$ evaluated at $(y_{n},\,y_{n-1},\,\ldots,\,y_{0})$.
\begin{proof}
We begin by proving \eqref{EEQ}. By combining \eqref{Vtaustar} and \eqref{EE} and setting $\A=\Q$ we obtain
\begin{align*}
    \text{EE}_\Q(t_n) &= \E_\Q\left[\E_\Q^n[\e^{-r(\tau_n(\boldsymbol{f}_n^*)-t_n)}g_{\tau_n(\boldsymbol{f}_n^*)}(X_{\tau_n(\boldsymbol{f}_n^*)})]\I_{\{\tau(\boldsymbol{f}^*)> t_n\}}\,|\,X_{t_0}=x_{t_0}\right]\\
    &=\E_\Q\left[\e^{-r(\tau_n(\boldsymbol{f}_n^*)-t_n)}g_{\tau_n(\boldsymbol{f}_n^*)}(X_{\tau_n(\boldsymbol{f}_n^*)})\I_{\{\tau(\boldsymbol{f}^*)> t_n\}}\,|\,X_{t_0}=x_{t_0}\right],
\end{align*}
where the final step is justified by law of total expectation. The expected 
exposure under the $\P-$measure can be rewritten in the following way
\begin{align}
    \text{EE}_\P(t_n) =& \E_\P\left[\E_\Q^n[\e^{-r(\tau_n(\boldsymbol{f}_n^*)-t_n)}g_{\tau_n(\boldsymbol{f}_n^*)}(X_{\tau_n(\boldsymbol{f}_n^*)})]\I_{\{\tau(\boldsymbol{f}^*)> t_n\}}\,|\,X_{t_0}=x_{t_0}\right]\nonumber\\
    =&\int_{(y_n,\ldots,y_1)\in\R^d\times\cdots\times\R^d}\E_\Q\left[\e^{-r(\tau_n(\boldsymbol{f}_n^*)-t_n)}g_{\tau_n(\boldsymbol{f}_n^*)}(X_{\tau_n(\boldsymbol{f}_n^*)})\,|X_{t_n}=y_n\,\right]\nonumber\\
    &\times \I_{\{\tau(\boldsymbol{f}^*)> t_n\}}p_{X_{t_1},\ldots,\,X_{t_n}}(y_n,\ldots,y_1)\d y_n\cdots\d y_1\nonumber\\
    =&\int_{(y_n,\ldots,y_1)\in\R^d\times\cdots\times\R^d}\E_\Q\left[\e^{-r(\tau_n(\boldsymbol{f}_n^*)-t_n)}g_{\tau_n(\boldsymbol{f}_n^*)}(X_{\tau_n(\boldsymbol{f}_n^*)})\,|X_{t_n}=y_n\,\right]\nonumber\\ \label{int_EE}
    &\times\I_{\{\tau(\boldsymbol{f}^*)> t_n\}}\frac{ p_{X_{t_n},\ldots,\,X_{t_1}}(y_n,\ldots,y_1)}{ q_{X_{t_n},\ldots,\,X_{t_1}}(y_n,\ldots,y_1)} q_{X_{t_n},\ldots,\,X_{t_1}}(y_n,\ldots,y_1)\d y_n\cdots\d y_1,
\end{align}
where $q_{X_{t_n},\ldots,\,X_{t_1}}(y_n,\ldots,y_1)$ and $p_{X_{t_n},\ldots,\,X_{t_1}}(y_n,\ldots,y_1)$ are joint densities of \\$X_{t_n},\ldots,X_{t_1}$ (conditioned on $X_{t_0}=x$) under the measures $\P$ and $\Q$, respectively. The fact that $\P$ and $\Q$ are equivalent, guarantees that the quotient in \eqref{int_EE} is well defined. Furthermore, due to the Markov property of $X$, we have \begin{equation}
    p_{X_{t_n},\ldots,\,X_{t_1}}(y_n,\ldots,y_1)=p_{X_{t_n}|X_{t_{n-1}}}(y_n|y_{n-1})\times\cdots\times p_{X_{t_1}|X_{t_{0}}}(y_1|y_{0}).\label{density_EE}
\end{equation}
The same argumentation holds for $q_{X_{t_n},\ldots,\,X_{t_1}}(y_n,\ldots,y_1)$. The proof is finalized by inserting the product of density functions
\eqref{density_EE} in \eqref{int_EE} and writing the expression as a $\Q$ expectation, and again, use the law of total expectation.
Is this proof too long? Is the statement obvious?
\end{proof}
\end{theorem}
Theorem \ref{EEtheorem} opens up for approximations of the expected exposure directly from the discounted cashflows. We now consider the special case, when $X$ is described by a diffusion-type stochastic differential equation (SDE). Let $\mu^{\Q}\colon[t_0,t_N]\times\R^d\to\R^d$ and $\sigma\colon[t_0,\,t_N]\times\R^d\to\R^d\times\R^d$ be the drift and diffusion coefficients, respectively, and let $(W^\Q_t)_{t\in[t_0,t_N]}$ be a $d$-dimensional standard Brownian motion under the measure $\Q$. We assume that $\mu^\Q$ and $\sigma$ satisfy the usual conditions (see \textit{e.g.}, \cite{oksendal}) for existence of a unique, strong solution $X$ to
\begin{equation*}
    \d X_t=\mu^\Q(t,\,X_t)\d t+\sigma(t,\,X_t)\d W^\Q_t,\ t\in[t_0,\,t_N];\quad X_{t_0}=x_{t_0}.
\end{equation*}
Let $\mu^{\P}\colon[t_0,t_N]\times\R^d\to\R^d$ and assume that $\sigma(t,\,X_t)$ is invertible and that for $t\in[t_0,t_N]$\footnote{ For a matrix $A\in\R^{d\times d}$, $\|A\|_\text{max}=\max_{ij}|A_{ij}|$.} \\$\|\sigma^{-1}(t,\,X_t)\|_\text{max}$ is bounded almost surely. Then, $(W^\P_t)_{t\in{[t_0,t_N]}}$ given by \begin{equation*}
    \d W_t^\P=-\sigma^{-1}(t,\,X_t)\left(\mu^\P(t,\,X_t)-\mu^\Q(t,\,X_t)\right)\d t + \d W_t^\Q\label{dWPdWQ}
\end{equation*}
is a Brownian motion under the measure $\P$. Furthermore, under the measure $\P$ it holds almost surely that $X$ is described by \begin{equation}\label{dXt}
    \d X_t=\mu^\P(t,\,X_t)\d t+\sigma(t,\,X_t)\d W^\P_t,\ t\in[t_0,\,t_N];\quad X_{t_0}=x_{t_0}.
\end{equation}
As a way to rewrite $\text{EE}_\P(t)$ as a conditional expectation under the measure $\Q$, we define a process $U=(U_t)_{t\in[t_0,t_N]}$, which follows the SDE\begin{equation}\label{dUt}
            \d U_{t}= \mu^{\P}(t,\,U_t)\d t+\sigma(t,\,U_t)\d W_t^\Q,\ t\in[t_0,t_N];\quad U_{t_0}=x_{t_0}.
\end{equation}
The reason for introducing this process is that $U$ has the same distribution under the measure $\Q$ as $X$ has under the measure $\P$.  
We can then express $\text{EE}_\P(t_n)$ with only $\Q-$expectations in the following way
\begin{align}
    \text{EE}_\P(t_n)& =\E_\P\big[V_{t_n}(X_{t_n})\I_{\{\tau(\boldsymbol{f}^*(X))> t_n\}}\,|\,X_{t_0}=x_{t_0}\big]\nonumber\\
    &=\E_\Q\big[V_{t_n}(U_{t_n})\I_{\{\tau\left(\boldsymbol{f}^*(U)\right)> t_n\}}\,|\,U_{t_0}=x_{t_0}\big]\nonumber\\
    &=\E_\Q\bigg[\E_\Q\left[\e^{-r\left(\tau_n(\boldsymbol{f}_n^*(X))-t_n\right)}g_{\tau_n(\boldsymbol{f}_n^*(X))}\left(X_{\tau_n(\boldsymbol{f}_n^*(X))}\right)\,|\,X_{t_n}=U_{t_n}\right]\times\I_{\{\tau(\boldsymbol{f}^*(U))> t_n\}}\,|\,U_{t_0}=x_{t_0}\bigg]\label{innerexp}.
\end{align}

\begin{remark}
 Regarding the equality between the right hand side of the first line and the second line, we want to emphasize that $X$ under the measure $\P$, and $U$ under the measure $\Q$ do not represent the same stochastic process (as they would have done after a change of measure). If they were, then the conditional expectation would change, and the equality would not hold. To enforce the equality to hold we could have corrected with the Radon--Nikodym derivative, and obtained \eqref{EEP}.
 However, to find a way to write $\mathrm{EE}_\P(t_n)$ without having to include the Radon--Nikodym derivative, we introduce another process U, which is distributed in such a way that the equality holds, \textit{i.e.,} the conditional expectation under $\P$, when using $X$, should equal the conditional expectation under $\Q$, when using $U$. For this to hold, it is sufficient that the distribution of $X$ under the measure $\P$ equals the distribution of $U$ under the measure $\Q$, which is satisfied when $X$ follows \eqref{dXt} and $U$ follows \eqref{dUt}.
 \end{remark}

\begin{remark}
The final equality is obtained by the fact that $V_{t_n}(U_{t_n})$ is the option value at the (random) state $(t_n,U_{t_n})$, given by \eqref{Vtaustar}.
\end{remark}

Before we get rid of the inner expectation in \eqref{innerexp}, we need to define a process following \eqref{dUt} on $[t_0,t_n]$, and the dynamics of \eqref{dXt} on $[t_n,t_N]$ with (stochastic) initial condition $X_{t_n}=U_{t_n}$. We denote such a process by $\tilde{X}^{t_n}=(\tilde{X}^{t_n}_t)_{t\in[t_0,t_N]}$,  and conclude that $\tilde{X}^{t_n}$ should satisfy the following SDE
\begin{equation}\label{XPQtn}
    \d \tilde{X}^{t_n}_t=\mu^{\P,\Q,t_n}(t,\,\tilde{X}^{t_n}_t)\d t+\sigma(t,\,\tilde{X}^{t_n}_t)\d W^\Q_t,\ t\in[t_0,\,t_N];\quad \tilde{X}^{t_n}_{t_0}=x_{t_0},
\end{equation}
where $\mu^{\P,\Q,t_n}(t,\cdot)=\mu^\P(t,\cdot)\I_{\{t\leq t_n\}}+\mu^\Q(t,\cdot)\I_{\{t>t_n\}}$. Note that we have implicitly assumed that also $\mu^\P$ satisfies the usual conditions for existence of a unique strong solution, $U$, to \eqref{dUt}. As a consequence of this assumption we are also guaranteed that there exists a unique strong solution, $\tilde{X}^{t_n}$, to \eqref{XPQtn}. We can then use the law of total expectation to obtain
\begin{align}
     \text{EE}_\P(t_n)&=\E_\Q\left[V_{t_n}(\tilde{X}^{t_n}_{t_n})\I_{\{\tau(\boldsymbol{f}^*)> t_n\}}\,|\,\tilde{X}^{t_n}_{t_0}=x_{t_0}\right]\nonumber\\\label{EEP2_anal}
     &= \E_\Q\left[\e^{-r(\tau_n(\boldsymbol{f}_n^*)-t_n)}g_{\tau_n(\boldsymbol{f}_n^*)}\left(\tilde{X}^{t_n}_{\tau_n(\boldsymbol{f}_n^*)}\right)\I_{\{\tau(\boldsymbol{f}^*)> t_n\}}\,|\,\tilde{X}^{t_n}_{t_0}=x_{t_0}\right],
\end{align}
where we remind ourselves that $\tau(\boldsymbol{f}^*)=\tau(\boldsymbol{f}^*(\tilde{X}^{t_n}))$ and $\tau_n(\boldsymbol{f}_n^*)=\tau_n(\boldsymbol{f}_n^*(\tilde{X}^{t_n}))$.

In the next sections we describe a method to approximate $\boldsymbol{f}^*(\cdot)$. It is then straightforward to estimate \eqref{EEQ}, \eqref{EEP} and \eqref{EEP2_anal} by Monte-Carlo sampling. Furthermore, in Section \ref{learn_pw_optionvalues}, we introduce a method to approximate the price function $V_{t_n}(\cdot)$, which makes it straightforward to also approximate the potential future exposure \eqref{PFE}.

\section{Learning stopping decisions}\label{learnSD}
In the first part of this section, we present the DOS algorithm, which was proposed in \cite{DOS}. The idea is to use fully connected neural networks to approximate the decision functions introduced in the previous section. The neural networks are optimized backwards in time with the objective to maximize the expected discounted cashflow at each exercise date. In the second part of this section, we suggest some adjustments that can be done in order to make the optimization more efficient.
\subsection{The Deep Optimal Stopping algorithm}\label{DOSalgo}
As indicated above, the core of the algorithm is to approximate decision functions. To be more precise, for $n\in\{0,\,1,\,\ldots,\,N\}$, the decision function $f_n$ is approximated by a fully connected neural network of the form $f^{\theta_n}_n\colon\mathbbm{R}^d\to\{0,\,1\}$, where $\theta_n\in\mathbbm{R}^{q_n}$ is a vector containing all the $q_n\in\N$ trainable\footnote{ Parameters that are subject to optimization.} parameters in network $n$. We assume that the initial state, $x_0\in\R
^d$, is such that it is sub-optimal to exercise the option at $t_0$, and therefore set $\theta_0$ such that $f_0^{\theta_0}(x_0)=0$ (for a further discussion, see Remark 6 in \cite{DOS}). Since binary decision functions are discontinuous, and therefore unsuitable for gradient-type optimization algorithms, we use as an intermediate step, the neural network $F^{\theta_n}_n\colon\mathbbm{R}^d\to(0,\,1)$. Instead of a binary decision, the output of the neural network $F^{\theta_n}_n$ can be viewed as the probability\footnote{ However the interpretation as a probability may be helpful, one should be careful since it is not a rigorous mathematical statement. It should be clear that there is nothing random about the stopping decisions, since the stopping time is $\F_t-$measurable. It can also be interpreted as a measure on how certain we can be that exercise is optimal.} for exercise to be optimal. This output is then mapped to 1 for values above (or equal to) 0.5, and to 0 otherwise, by defining $f^{\theta_n}_n(\cdot)=\mathfrak{a}\circ F_n^{\theta_n}(\cdot)$, where $\mathfrak{a}(x)=\mathbbm{I}_{\{x\geq1/2\}}$. Our objective is to find $\theta_n$ such that \begin{equation}\label{theoretical_train}
    \E_\Q^n\left[f^{\theta_n}_n(X_{t_n})g_{t_n}(X_{t_n})+(1-f^{\theta_n}_n(X_{t_n}))\e^{-r(\tau_{n+1}(\boldsymbol{f}_{n+1}^*)-t_n)}g_{\tau_{n+1}(\boldsymbol{f}_{n+1}^*)}\left(X_{\tau_{n+1}(\boldsymbol{f}_{n+1}^*)}\right)\right],
\end{equation}
is as close as possible to $V_{t_n}(X_{t_n})$ (in mean squared sense), where $\boldsymbol{f}_{n+1}^*$ is the vector of optimal decision functions, defined in \eqref{opt_dec}. Although \eqref{theoretical_train} is an accurate representation of the optimization problem, it gives us some practical problems. In general, we have no access to either $\boldsymbol{f}_{n+1}^*$ or the distribution of $V_{t_n}(X_{t_n})$ and in most cases the expectation needs to be approximated. We however notice that at maturity, the option value is equal to its intrinsic value, \textit{i.e.}, $V_{t_N}(\cdot)\equiv g_{t_N}(\cdot)$, which implies that $f_N^*\equiv1$ and $\tau_N(\boldsymbol{f}_N^*)=t_N$. With this insight, we can write \eqref{theoretical_train} with $n=N-1$ in the form
\begin{equation}\label{last_step_theoretical}
    \E_\Q^{N-1}\left[f^{\theta_{N-1}}_{N-1}(X_{t_{N-1}})g_{t_{N-1}}(X_{t_{N-1}})+(1-f^{\theta_{N-1}}_{N-1}(X_{t_{N-1}}))\e^{-r(t_N-t_{N-1})}g_{t_N}\left(X_{t_N}\right)\right],\end{equation}
which can be approximated by Monte-Carlo sampling. Given $M\in\N$ samples, distributed as $X$, which for $m\in\{1,2,\ldots,M\}$ is denoted by $x=(x_{t_n}(m))_{n=0}^N$, we can approximate \eqref{last_step_theoretical} by
\begin{equation}\label{last_step_empirical}
    \frac{1}{M}\sum_{m=1}^{M}\left(f^{\theta_{N-1}}_{N-1}(x_{t_{N-1}}(m))g_{t_{N-1}}(x_{t_{N-1}}(m))+(1-f^{\theta_{N-1}}_{N-1}(x_{t_{N-1}}(m)))\e^{-r(t_N-t_{N-1})}g_{t_N}\left(x_{t_N}(m)\right)\right).\end{equation}
Note that the only unknown entity in \eqref{last_step_empirical} is the parameter $\theta_{N-1}$ in the decision function $f_{N-1}^{\theta_{N-1}}$. Furthermore, we note that we want to find $\theta_{N-1}$ such that \eqref{last_step_empirical} is maximized, since it represents the average cashflow in $[t_{N-1},t_N]$. Once $\theta_{N-1}$ is optimized, we use this parameter and find $\theta_{N-2}$ such that the average cashflow on $[t_{N-2},t_N]$ is maximized.

In the next section, we explain the parameters $\theta_n$ and present the structure for the neural networks used in this paper. 

\subsubsection{Specification of the neural networks used}\label{NN_spec}
For completeness, we introduce all the trainable parameters that are contained in each of the parameters $\theta_1,\theta_2,\ldots,\theta_{N-1}$, and present the structure of the networks.
\begin{itemize}
    \item We denote the dimension of the input layers by $\mathfrak{D}^{\text{input}}\in\N$, and we assume the same input dimension for all $n\in\{1,2,\ldots,N-1\}$ networks. The input is assumed to be the market state $x_{t_n}^{\text{train}}\in\R^d$, and hence $\mathfrak{D}^{\text{input}}=d$. However, we can add additional information to the input that is mathematically redundant but helps the training, \textit{e.g.,} the immediate pay-off, to obtain as input $\left(\text{vec}(x_{t_n}^\text{train}(m)),\,g_{t_n}\left(x_{t_n}^\text{train}(m)\right)\right)\in\R^{d+1}$, which would give $\mathfrak{D}^{\text{input}}=d+1$. In \cite{DOS}, the authors claim that by adding the immediate pay-off to the input, the efficiency of the algorithm was improved;
    
    \item For network $n\in\{1,2,\ldots,N-1\}$, we denote the number of layers\footnote{ Input and output layers included.} by $\mathfrak{L}_n\in\N$, and for layer $\ell\in\{1,2,\ldots,\mathfrak{L}_n\}$, the number of nodes by $\mathfrak{N}_{n,\ell}\in\N$. Note that $\mathfrak{N}_{n,1}=\mathfrak{D}^{\text{input}}$;
    
    \item For network $n\in\{1,2,\ldots,N\}$, and layer $\ell\in\{2,3,\ldots,\mathfrak{L}_n\}$ we denote the weight matrix, acting between layers $\ell-1$ and $\ell$, by $w_{n,\ell}\in\R^{\mathfrak{N}_{n,\ell-1}\times\mathfrak{N}_{n,\ell}}$, and the bias vector by $b_{n,\ell}\in\R^{\ell}$;
    
    \item For network $n\in\{1,2,\ldots,N\}$, and layer $\ell\in\{2,3,\ldots,\mathfrak{L}_n\}$ we denote the (scalar) activation function by $a_{n,\ell}\colon\R\to\R$ and the vector activation function by $\boldsymbol{a}_{n,\ell}\colon\R^{\mathfrak{N}_{n,\ell}}\to\R^{\mathfrak{N}_{n,\ell}}$, which for $x=(x_1,x_2,\ldots,x_{\mathfrak{N}_{n,\ell}})$ is defined by 
    \begin{equation*}
        \boldsymbol{a}_{n,\ell}(x)=\begin{pmatrix}a_{n,\ell}(x_1)\\
        \vdots\\
        a_{n,\ell}(x_{\mathfrak{N}_{n,\ell}})\end{pmatrix};
    \end{equation*}
    \item The output of our network should belong to $(0,1)\subset\R$, meaning that the output dimension of our neural network, denoted by $\mathfrak{D}^{\text{output}}$ should equal 1. To enforce the output to only take on values in $(0,1)$, we restrict ourselves to activation functions of the form $a_{n,\mathfrak{L}_{n}}\colon\R\to(0,1)$.
\end{itemize}
Network $n\in\{1,2,\ldots\,N-1\}$ is then defined by \begin{equation}
    F_n^{\theta_n}(\cdot) = L_{n,\mathfrak{L}_n}\circ L_{n,\mathfrak{L}_n-1}\circ\cdots\circ L_{n,1}(\cdot),
\end{equation}
where for $n\in\{1,2,\ldots,N-1\}$ and for $x\in\R^{\mathfrak{L}_{n,\ell-1}}$, the layers are defined as 
\begin{equation*}
L_{n,\ell}(x)=\begin{cases}x,&\text{for }\ell=1,\\ 
    \boldsymbol{a}_{n,\ell}(w_{n,\ell}^Tx+b_{n,\ell}),&\text{for }\ell\geq 2,\end{cases}
\end{equation*}
where $w_{n,\ell}^T$ is the matrix transpose of $w_{n,\ell}$. The trainable parameters of network $n\in\{1,2,\ldots,N-1\}$ are then given by the list\begin{equation*}
    \theta_n=\left\{w_{n,2},b_{n,2},w_{n,3},b_{n,3},\ldots,w_{n,\mathfrak{L}_n}, b_{n,\mathfrak{L}_n}\right\}.
\end{equation*}
Furthermore, since we have $N-1$ neural networks, we denote by \begin{equation*}
    \boldsymbol{\theta}_n=\{\theta_n,\theta_{n+1},\ldots,\theta_{N-1}\}
\end{equation*}
the trainable parameters in the neural networks at exercise dates $\mathbbm{T}_n$ and by $\boldsymbol{\theta}=\boldsymbol{\theta}_1$.

\subsubsection{Training and valuation}\label{train_NN}
The main idea of the training and valuation procedure is to fit the parameters to some training data, and then use the fitted parameters to make informed decisions with respect to some unseen, so-called, valuation data. 

The training part of the algorithm is summarized below in pseudo code.   
\newline
\newline
\textbf{Training phase}:\newline
Sample $M_{\text{train}}\in\mathbbm{N}$ independent realizations of $X$, which for $m\in\{1,\,2,\,\ldots,\,M_{\text{train}}\}$ are denoted  $(x_{t_n}^{\text{train}}(m))_{n=0}^{N}$. At maturity, define the cashflow as $\text{CF}_N(m)=g_{t_N}(x_{t_N}^{\text{train}}(m))$. \newline
For $n=N-1,\,N-2,\,\ldots,\,1$, do the following:
\begin{enumerate}
    \item\label{loss_emp} Find a $\hat{\theta}_{n}\in\R^{q_{n}}$ which approximates
\begin{align*} \hat{\theta}_n^*\in&\argmax_{\theta\in\mathbbm{R}^{q_n}}\bigg(\frac{1}{M_\text{train}}\sum_{m=1}^{M_\text{train}}F_n^{\theta}\big(x_{t_n}^\text{train}(m)\big)g_{t_n}\big(x_{t_n}^\text{train}(m)\big)\\
&+\left(1-F_n^{\theta}\big(x_{t_n}^\text{train}(m)\big)\right)\e^{-r(t_{n+1}-t_n)}\text{CF}_{n+1}(m)\bigg).\end{align*}
In machine learning terminology, this would give an (empirical) loss-function of the form
\begin{equation*}
    L(\theta;x^{\text{train}}) =- \frac{1}{M_\text{train}}\sum_{m=1}^{M_\text{train}}F_n^{\theta}\left(x_{t_n}^\text{train}(m)\right)g_{t_n}\left(x_{t_n}^\text{train}(m)\right)+\left(1-F_n^{\theta}\left(x_{t_n}^\text{train}(m)\right)\right)\e^{-r(t_{n+1}-t_n)}\text{CF}_{n+1}(m).
\end{equation*}
The minus sign in the loss-function transforms the problem from a maximization to minimization, which is the standard formulation in the machine learning community. Note the straightforward relationship between the loss function and the average cashflows in \eqref{last_step_empirical}. In practice, the data is often divided into mini-batches, for which the loss-function is minimized consecutively. 
\item For $m=1,\,2,\,\ldots,\,M_{\text{train}}$, update the discounted cashflow according to: \begin{align*}
    \text{CF}&_{n}(m)\\
    =&f_n^{\hat{\theta}_{n}}\big(x_{t_n}^\text{train}(m)\big)g_{t_n}\big(x_{t_n}^\text{train}(m)\big)+\left(1-f_n^{\hat{\theta}_{n}}\big(x_{t_n}^\text{train}(m)\big)\right)\e^{-r(t_{n+1}-t_n)}\text{CF}_{n+1}(m).
    \end{align*}
\end{enumerate}
The performance of the algorithm is not particularly sensitive to the specific choice of the number of hidden layers, number of nodes, optimization algorithm, etc. Below is a list of the most relevant parameters/structural choices:
\begin{itemize}
    \item Initialization of the trainable parameters, where a typical procedure is to initialize the biases to 0, and sample the weights independently from a normal distribution; 
    \item The activation functions $a_{\ell,n}$, which are used to add a non-linear structure to the neural networks. In our case we have the strict requirement that the activation function of the output layer maps $\R$ to $(0,1)$. This could, however, be relaxed as long as the activation function is both upper and lower bounded, since we can always scale and shift such output to take on values only in $(0,1)$. For a discussion of different activation functions, see \textit{e.g.,} \cite{activation_functions}.;
    \item The batch size, $B_{n}\in\{1,2,\ldots,M_\text{train}\}$, is the number of training samples used for each update of $\theta_n$, \textit{i.e.,} with $B_n=M_\text{train}$, the loss function is of the form defined in step \ref{loss_emp} above. Note that if we want all batches to be of equal size, we need to choose $B_{n}$ to be a multiplier of $M_\text{train}$;
    \item For each update of $\theta_n$, we use an optimization algorithm, for which a common choice is the Adam optimizer, proposed in \cite{Adam}. Depending on the choice of optimization algorithm, there are different parameters related to the specific algorithm to be chosen. One example is the so-called learning rate which decides how much the parameter, $\theta_n$ is adjusted after each batch.
\end{itemize}
Once the parameters, $\{\theta_1,\theta_2,\ldots,\theta_{N-1}\}$, have been optimized we can use the algorithm for valuation.
\newline\newline
\textbf{Valuation phase}:\newline
Sample $M_\text{val}\in\mathbbm{N}$ independent realizations of $X$, denoted $\left(x_{t_n}^\text{val}(m)\right)_{n=0}^N$. We emphasize that the valuation data should be independent from the training data. 
 Denote the vector of decision functions by \begin{equation*}
    \boldsymbol{f}_n^{\boldsymbol{\hat{\theta}}}=\left(f^{\hat{\theta}_n}_n,\,f^{\hat{\theta}_{n+1}}_{n+1},\ldots,f^{\hat{\theta}_{N-1}}_{N-1}\right),\end{equation*}
    and $\boldsymbol{f}^{\boldsymbol{\hat{\theta}}}=\boldsymbol{f}^{\boldsymbol{\hat{\theta}}}_0$. We then obtain for sample $m$, \textit{i.e.,} $x^{\text{val}}(m)$, the following stopping rule at time $t_n$
\begin{equation*}
    \label{taustar}
    \tau_n\left(\boldsymbol{f}_n^{\boldsymbol{\hat{\theta}}}\left(x^{\text{val}}(m)\right) \right) = \sum_{m=n}^Nt_mf^{\hat{\theta}_m}_m\left(x^\text{val}_{t_m}(m)\right)\prod_{j=0}^{m-1}\left(1-f^{\hat{\theta}_j}_j\left(x^\text{val}_{t_j}(m)\right)\right).
\end{equation*} The estimated option value at $t_0$ is then given by\begin{equation}\label{V0_est}
\hat{V}_{t_0}(x_0) = \frac{1}{M_{\text{val}}}\sum_{m=1}^{M_{\text{val}}}\e^{-r\left(\tau\left(\boldsymbol{f}^{\hat{\theta}}\right) - t_0\right)}g_{\tau\left(\boldsymbol{f}^{\boldsymbol{\hat{\theta}}}\right)}\left(x^\text{val}_{\tau(\boldsymbol{f}^{\boldsymbol{\hat{\theta}}})}(m)\right),\end{equation} where we recall that $\tau(\boldsymbol{f}^{\boldsymbol{\hat{\theta}}})=\tau_0\left(\boldsymbol{f}_0^{\boldsymbol{\hat{\theta}}}\left(x^{\text{val}}(m)\right) \right)$. Note that, by construction, any stopping strategy is sub-optimal, implying that the estimate \eqref{V0_est} is biased low. It should be pointed out that it is possible to derive a biased high estimate of \eqref{V0} from a dual formulation of the optimal stopping problem, which is described in \cite{DOS}. In addition, numerical results in \cite{DOS} show a tight interval for the biased low and biased high estimates for a wide range of different problems.

\subsection{Proposed adjustments to the algorithm}
The presentation of the DOS algorithm in \cite{DOS} is in a general form. In addition to the pricing of Bermudan options, the authors considered the non-Markovian problem to optimally stop a fractional Brownian motion (this is done by including also the historical states in the current state of the system). Since the aim of this paper is more specific (to approximate exposure profiles of Bermudan options), it is natural to use more of the known underlying structure of this specific problem. In this section we use some properties of the specific problems, and propose some adjustments to the DOS-algorithm, which make the training procedure more efficient. 
\subsubsection{Reuse of neural network parameters}
The first proposed adjustment is to reuse parameters of neural networks that have already been optimized. We note that for a single Bermudan option (possibly with a high-dimensional underlying asset) the pay-off functions are identical at all exercise dates, \textit{i.e.,} $g_{t_n}=g_{t_m}$ for all $t_n,t_m\in\mathbbm{T}$. In this case the stopping rules at adjacent exercise dates are similar, especially when $t_{n+1}-t_{n}$ is small. We therefore use the stopping strategy at $t_{n+1}$ as an ''initial guess'' for the stopping strategy at $t_n$. This is done by initializing the trainable parameters in network $n$ by the already optimized parameters in network $n+1$, \textit{i.e.,} at $t_n$, initialize $\theta_n$ by $\hat{\theta}_{n+1}$. This allows us to use smaller learning rates leading to a more efficient algorithm. 

\subsubsection{Use simple stopping decisions when possible}
The term simple stopping decisions is loosely defined as stopping decisions that follow directly without any sophisticated optimization algorithm. The most obvious example is when the contract is out-of-the-money, in which case it is never optimal to exercise. For $t_n\in\mathbbm{T}$, we define the set of in-the-money points and out-of-the-money points as \begin{equation*}
\text{ITM}_n=\{x\in\R^d\,|\,g_{t_n}(x)>0\},\quad \text{OTM}_n=\{x\in\R^d\,|\,g_{t_n}(x)=0\},
\end{equation*}
respectively. Another, less obvious insight is that, given a single Bermudan option with identical pay-off functions at all exercise dates, if it is optimal to exercise at $(t_n,x)$, then it is also optimal to exercise at $(t_{n+1},x)$. Or in other words, the exercise region is non-decreasing with time. This statement is formulated as a theorem below.
\begin{theorem}\label{increaseExReg}Define the set of exercise dates by $\{t_0,\ldots,t_n,t_{n+1},\ldots,t_N,\,t_{N}+\Delta \}$, and let $\Delta =t_{n+1}-t_n=t_{N+1}-t_{N}\geq 0$. Note that an equidistant time grid is sufficient, but not necessary for the above to be satisfied. Moreover, assume that\begin{equation*}
    V_{t_n}(\cdot\,;\,t_N)=V_{t_{n+1}}(\cdot\,;t_N+\Delta ),
\end{equation*}
where $t_N$ and $t_N+\Delta $ indicate the maturity of otherwise identical contracts of the form \eqref{Vn}, with $g=g_{t_n}$ for all exercise dates $t_n$. 
Then, for any $\bar{x}\in\mathcal{E}_{n}$, it holds that $\bar{x}\in\mathcal{E}_{n+1}$.
\end{theorem}
\begin{proof}
Since $\bar{x}\in\mathcal{E}_n$, $V_{t_n}(\bar{x}\,;\,t_N)=g(\bar{x})$ and $V_{t_{n}}(\bar{x}\,;\,t_N)=V_{t_{n+1}}(\bar{x}\,;\,t_N+\Delta )$ we have that $V_{t_{n+1}}(\bar{x}\,;\,t_N+\Delta )=g(\bar{x})$. From \eqref{Vn} we also see that $V_{t_{n+1}}(\bar{x}\,;\,t_N)\leq V_{t_{n+1}}(\bar{x}\,;\,t_N+\Delta )$ and $V_{t_{n+1}}(\bar{x}\,;\,t_N)\geq g(\bar{x})$. Therefore $V_{t_{n+1}}(\bar{x}\,;\,t_N)= g(\bar{x})$ and $\bar{x}\in\mathcal{E}_{n+1}$. 
\end{proof} 
Theorem \ref{increaseExReg} shows that the exercise region is non-decreasing with time, but since the optimization of the neural network parameters is carried out backwards in time we instead use the fact that the continuation region is non-increasing with time. In practice, this leads to the following three alternatives:
\begin{enumerate}
    \item[A1.] Use all training data in the optimization algorithm (as the algorithm is described in Subsection \ref{DOSalgo});\\
    \item[A2.] At $t_n\in\mathbbm{T}$, use the subset of the training data satisfying $x^\text{val}_{t_n}(m)\in\text{ITM}_n$ in the optimization algorithm. Define the decision functions as \begin{equation*}
        f_n^{\hat{\theta}_n}(\cdot)=\mathbbm{I}_{\{g_{t_n}(\cdot)>0\}}(\mathfrak{a}\circ F_n^{\hat{\theta}_n}(\cdot)).\end{equation*}
    To only use ''in the money paths'' is also employed in the Least Squares Method (LSM), proposed in \cite{LSM};
    \item[A3.] At $t_n$ use the subset of the training data $x^\text{train}_{t_n}(m)\in \mathcal{E}_{n+1}$ in the optimization algorithm. Define the decision functions as \begin{equation*}
        f_n^{\hat{\theta}_n}(\cdot)=f_{n+1}^{\hat{\theta}_{n+1}}(\cdot)(\mathfrak{a}\circ F_n^{\hat{\theta}_n}(\cdot)).  \end{equation*}
\end{enumerate}
In Figure \ref{train_points} the three cases above are visualized for a two-dimensional max call option at one of the exercise dates. To the left we have the blue points belonging to $\mathcal{E}_{n+1}$ (used for optimization in A3), the blue and red points belong to $\text{ITM}_n$ (used for optimization in A2) and the blue, red and yellow points are all the available data (used for optimization in A1). To the right, we see the fraction of the total data used in each case at each exercise date.
    \begin{figure}[htp]
  \centering
  \begin{tabular}{ccc}
         \includegraphics[width=82mm]{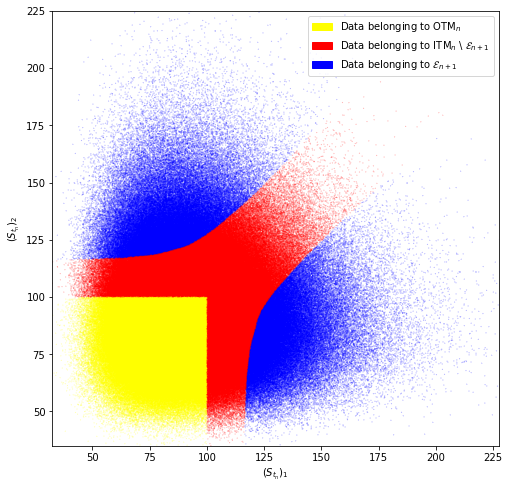}  & \includegraphics[width=82mm]{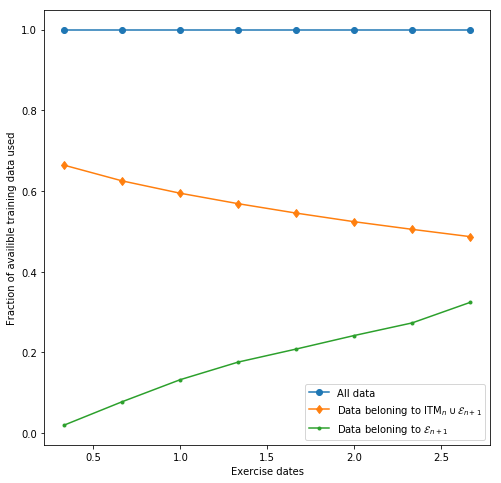}\\
  \end{tabular}
  \caption{\textbf{Left:} Blue points in $\mathcal{E}_{n+1}$ (used for optimization in A3), blue and red points  in $\text{ITM}_n$ (used for optimization in A2) and the blue, red and yellow points are all the available data (used for optimization in A1). \textbf{Right:} The fraction of the total data used in each case at each exercise date.}
  \label{train_points}
  \end{figure}

\section{Learning pathwise option values}\label{learn_pw_optionvalues}
In Section \ref{learnSD} an algorithm to learn stopping decisions was described and \eqref{V0_est} gives an approximation of the option value at time $t_0$, given some deterministic initial state $X_{t_0}=x_{t_0}\in\R^d$. As described in Subsection \ref{EP_sec}, to compute exposure profiles we sometimes need additional information about the future distribution of the option values. In this section, we present two methods to approximate the pathwise option values at all exercise dates. The first method is the well-established Ordinary Least Squares (OLS) regression and the second method is a neural network-based least squares regression. Both methods rely on projections of conditional expectations on a finite-dimensional function space.

\subsection{Formulation of regression problem}\label{reg_prob}
Central for the regression algorithms presented in this section is the cashflow process, $Y=(Y_{t_n})_{n=0}^{N}$, defined as\footnote{Note that $Y$ is the discounted cashflow-process, which in the training phase was denoted by $\text{CF}_n$. The reason for using $Y_{t_n}$ instead in this section is that we want to emphasize that the pathwise valuation problem is, in fact, a standard regression problem in which $X$ and $Y$ usually are used to represent the observation vector, and the response variable, respectively.} \begin{equation}\label{Yn}
    Y_{t_n}=\e^{-r(\tau_n(\boldsymbol{f}^*_n)-t_n)}g_{\tau_n(\boldsymbol{f}^*_n)}\left(X_{\tau_n(\boldsymbol{f}^*_n)}\right),
\end{equation}
where $\tau_n(\cdot)$ and $\boldsymbol{f}^*$ are defined in \eqref{tauk} and \eqref{opt_dec}, respectively. We assume that for $t_n\in\mathbbm{T}$ it holds that \begin{equation*}
    \E_\Q[g_{t_n}(X_{t_n})^2]<\infty,
\end{equation*}
which also implies that $\E_\Q[Y_{t_n}^2]<\infty$.
The following theorem states that the option value, at some $t_n\in\mathbbm{T}$, is equivalent (in $L_2$ sense) to the so-called regression function. Furthermore, we see that the regression function can be obtained by solving a minimization problem. 

\begin{theorem}
Let $Y_{t_n}$ be as defined in \eqref{Yn} and for $h_n\in\D(\R^d;\,\R)$, define the so-called $L_2$-risk as
\begin{equation*}
    \E_\Q\left[|h_n(X_{t_n})-Y_{t_n}|^2\right].
\end{equation*}
It then holds that
\begin{equation*}
    \E_\Q\left[|V_{t_n}(X_{t_n})-Y_{t_n}|^2\right]=\min_{h_n\in\D(\R^d;\,\R)}\E_\Q\left[|h_n(X_{t_n})-Y_{t_n}|^2\right],
\end{equation*}
or equivalently \begin{equation*}
    V_{t_n}(\cdot)\in\argmin_{h_n\in\D(\R^d;\,\R)}\E_\Q\left[|h_n(X_{t_n})-Y_{t_n}|^2\right].
\end{equation*}
\end{theorem}

\begin{proof}
Define $m_n(x)=\E_\Q\left[Y_{t_n}\,|\,X_{t_n}=x\right]$. For an arbitrary function, $v\colon\R^d\to\R$, it holds that \begin{align*}
    \E_\Q\left[|v(X_{t_n})-Y_{t_n}|^2\right]=&\E_\Q\left[|v(X_{t_n})-m_n(X_{t_n})+m_n(X_{t_n})-Y_{t_n}|^2\right]\\
    =&\E_\Q\left[|v(X_{t_n})-m_n(X_{t_n})|^2\right]+\E_\Q\left[|m_n(X_{t_n})-Y_{t_n}|^2\right]\\
    &+2\E_\Q\left[(v(X_{t_n})-m_n(X_{t_n}))(m_n(X_{t_n})-Y_{t_n})\right].
\end{align*}
By the law of total expectation, the last term satisfies \begin{align*}
    \E_\Q[\left(v(X_{t_n})-m_n(X_{t_n})\right)&\left(m_n(X_{t_n})-Y_{t_n}\right)]\\
    =&\E_\Q\left[\E_\Q^n\left[(v(X_{t_n})-m_n(X_{t_n}))(m_n(X_{t_n})-Y_{t_n})\right]\right]\\
    =&\E_\Q\left[(v(X_{t_n})-m_n(X_{t_n}))(m_n(X_{t_n})-\E_\Q^n\left[Y_{t_n}\right])\right]=0,
\end{align*}
since $\E_\Q^n\left[Y_{t_n}\right]=m_n(X_{t_n})$ by definition of $Y_{t_n}$. This means that \begin{equation}\label{split}
    \E_\Q[|v(X_{t_n})-Y_{t_n}|^2]
    =\E_\Q\left[|v(X_{t_n})-m_n(X_{t_n})|^2\right]+\E_\Q\left[|m_n(X_{t_n})-Y_{t_n}|^2\right],
\end{equation}
which is clearly minimized when $v\equiv m_n$. Also, notice that $m_n(X_{t_n})=V^*_{t_n}(X_{t_n})$, and by Theorem \ref{VstarV}, \begin{equation*}
    V^*_{t_n}(X_{t_n})=V_{t_n}(X_{t_n}),
\end{equation*}
which concludes the proof.
\newline\newline

OR:\newline\newline
This is a straight forward consequence of the fact that the conditional expectation is the (least-squares) projection onto the Markov states, $h_n(X_{t_n})$. 
\end{proof}

In practice, the distribution of $(X,Y)$ is usually unknown. On the other hand, we are often able to generate samples distributed as\footnote{ In fact, we can only generate samples distributed as $(X,\hat{Y})$, where $\hat{Y}$ is the approximate discounted cashflow process obtained by using the neural network-based decision functions instead of the optimal decision functions in \eqref{Yn}. We give a short explanation of how this affects the regression in the end of this section.} $(X,Y)$. We consider some $t_n\in\mathbbm{T}$, and generate $M_{\text{reg}}\in\N$ independent realizations of the regression pair $(X_{t_n},Y_{t_n})$, which we denote by $\left(x^\text{reg}_{t_n}(m),y_{t_n}^\text{reg}(m)\right)_{m=1}^{M_\text{reg}}$. Similarly, we define the empirical $L_2$-risk by \begin{equation*}\label{empiricalL2}
    \frac{1}{M_\text{reg}}\sum_{m=1}^{M_\text{reg}}|h_n(x_{t_n}^\text{reg}(m))-y_{t_n}^\text{reg}(m)|^2.
\end{equation*}
With a fixed sample of regression pairs it is possible to find a function $h\in\D(\R^d;\,\R)$ such that the empirical $L_2$-risk equals zero. However, such a function is not a consistent estimator in general. Therefore, we want to use a smaller class of more regular functions. When choosing the function class, which we denote by $\mathcal{A}_M$, we need to keep two aspects in mind;
\begin{enumerate}
    \item[P1.] It should be ''rich enough'' to be able to approximate $V_{t_n}(\cdot)$ sufficiently accurately,\label{point1}
    \item[P2.] It should not be ''too rich'' since that may cause the empirical $L_2$-risk being an inaccurate approximation of the $L_2$-risk. Since this problem is more severe for smaller $M_\text{reg}$, it is reasonable to have the sample size in mind when choosing the function class, and hence the subscript $M$ on $\mathcal{A}_{M}$, where ''$\text{reg}$'' is dropped for notational convenience. A too rich function class may lead to what is known as overfitting in the machine learning community.\label{point2}
\end{enumerate}
Given a sample and a function class $\mathcal{A}_M$, we define the empirical regression function as
\begin{equation*}
    m_M(\cdot) \in \argmin_{h\in\mathcal{A}_M}\frac{1}{M_{\text{reg}}}\sum_{m=1}^{M_\text{reg}}\left|h\left(x_{t_n}^\text{reg}(m)\right)-y_{t_n}^\text{reg}(m)\right|^2.
\end{equation*}
Since \eqref{split} holds for arbitrary $v$, we can write the $L_2$-risk of the empirical regression function and the option value as 
\begin{equation*}\label{L2riskMkV}
    \E_\Q\left[|m_M(X_{t_n})-V_{t_n}(X_{t_n})|^2\right]=\E_\Q\left[|m_M(X_{t_n})-Y_{t_n}|^2\right]-\E_\Q\left[|V_{t_n}(X_{t_n})-Y_{t_n}|^2\right].
\end{equation*}
This in turn can be written in terms of the so-called estimation error (first term) and the approximation error (second term), \textit{i.e.,}\begin{align}\begin{split}\label{error_Y}
     \E_\Q[|m_M(X_{t_n})-&V_{t_n}(X_{t_n})|^2]\\
     =&\left(\E_\Q\left[|m_M(X_{t_n})-Y_{t_n}|^2\right]-\min_{h\in\mathcal{A}_M}\E_\Q[|h_n(X_{t_n})-Y_{t_n}|^2]\right)\\
     &+\left(\min_{h\in\mathcal{A}_M}\E_\Q[|h_n(X_{t_n})-Y_{t_n}|^2]-\E_\Q\left[|V_{t_n}(X_{t_n})-Y_{t_n}|^2\right]\right)\end{split}
\end{align}
The approximation error measures how well the option value can be estimated by functions in $\mathcal{A}_M$, which corresponds to (P1) above. The estimation error is the difference between the $L_2$-risk of the estimator $m_M$ and the optimal $h$ in $\mathcal{A}_M$, which corresponds to (P2) above. 

There is however, one problem with the approximation error above; we have assumed that we can sample realizations of $(X,Y)$, while we in practice only are able to sample from $(X,\hat{Y})$, with $\hat{Y}=(\hat{Y}_{t})_{t\in{[t_0,t_N]}}$ given by  \begin{equation*}
\label{hatYn}
    \hat{Y}_{t_n}=\e^{-r\left(\tau_n\left(\boldsymbol{f}^{\boldsymbol{\hat{\theta}}}_n\right)-t_n\right)}g_{\tau_n\left(\boldsymbol{f}^{\boldsymbol{\hat{\theta}}}_n\right)}\left(X_{\tau_n\left(\boldsymbol{f}^{\boldsymbol{\hat{\theta}}}_n\right)}\right).
\end{equation*}
By also taking into account that the regression is carried out against an approximation of $Y$, \eqref{error_Y} becomes instead
\begin{align}\begin{split}\label{error_hatY}
     \E_\Q[|m_M(X_{t_n})-&V_{t_n}(X_{t_n})|^2]\\
     \leq&\left(E_\Q\left[|m_M(X_{t_n})-\hat{Y}_{t_n}|^2\right]-\min_{h\in\mathcal{A}_M}\E_\Q[|h_n(X_{t_n})-\hat{Y}_{t_n}|^2]\right)\\
     &+\left(\min_{h\in\mathcal{A}_M}\E_\Q[|h_n(X_{t_n})-Y_{t_n}|^2]-\E_\Q\left[|V_{t_n}(X_{t_n})-Y_{t_n}|^2\right]\right)\\
     &+\left(\min_{h\in\mathcal{A}_M}\E_\Q[|h_n(X_{t_n})-\hat{Y}_{t_n}|^2]-\min_{h\in\mathcal{A}_M}\E_\Q[|h_n(X_{t_n})-Y_{t_n}|^2]\right)\\
     &+\E_\Q[|\hat{Y}_{t_n}-Y_{t_n}|^2].\end{split}
\end{align}
The first two lines in \eqref{error_hatY} are, again, the estimation error and the approximation error, respectively. The third line represents the difference between how well a function in $\mathcal{A}_M$ can approximate $\hat{Y}_{t_n}$ and $Y_{t_n}$ and the final row is the $L_2$-risk of our approximation of the discounted cashflow and the true discounted cashflow. Furthermore, note that the equality in \eqref{error_Y} has changed to an inequality in \eqref{error_hatY}.

In the next section, we introduce the two different types of function classes that are used in this paper.

\subsection{Ordinary least squares regression}\label{OLSsec}
At $t_n\in\mathbbm{T}$, we assume that $V_{t_n}(X_{t_n})$ can be represented by a linear combination of a countable set of $\F_n$-measurable basis functions. We denote by $\{\phi_b\}_{b=0}^\infty$ the basis functions and given optimal parameters $\alpha_{t_n}^{(1)},\,\alpha_{t_n}^{(2)},\ldots$ (in the sense that the $L_2$-risk against $V_{t_n}(X_{t_n})$ is minimized) and define \begin{equation*}\label{vinf}
    v(t_n,X_{t_n})=\sum_{b=1}^\infty\alpha_{t_n}^{(b)}\phi_b(X_{t_n}).
\end{equation*}
For practical purposes we use the first $B\in\N$ basis functions, so that
\begin{equation}\label{vB}
    v_B(t_n,X_{t_n})=\sum_{b=1}^B\alpha_{t_n}^{(b)}\phi_b(X_{t_n}).
\end{equation}
We now want to estimate \eqref{vB} by projecting a sample of realizations of $(X_{t_n},\,Y_{t_n})$ onto the $B$ first basis functions. This procedure is similar to LSM, \cite{LSM}. In the LSM, only ITM samples are used, which is motivated by the fact that it is never optimal to exercise an option that is OTM and the objective is to find the optimal exercise strategy. Furthermore, the authors claim that the number of basis functions needed to obtain an accurate approximation is significantly reduced since the approximation region is reduced by only considering ITM paths. However, this is not possible in our case since we need to approximate the option everywhere\footnote{ By ''everywhere'' we mean the region in which the distribution of $X_{t_n}$ has positive density. Of course, this is in many cases $\R^d$, so in practice by ''everywhere'' we mean the region in which the density is significantly positive.}. On the other hand, the exercise region, $\mathcal{E}_{n}$, has already been approximated (as described in Subsection \ref{DOSalgo}) and the option value in the exercise region is always known. This means that, similar to the LSM, the approximation region can be reduced (in many cases significantly) by only considering samples belonging to the continuation region, $\mathcal{C}_n$. 

Given a sample of regression pairs $\left(x_{t_n}^\text{reg}(m),\,y_{t_n}^\text{reg}(m)\right)_{m=1}^{M_\text{reg}}$, we let $M_\text{reg}^{\mathcal{C}_n}$ denote the number of samples belonging to $\mathcal{C}_n$ and let $\left(x_{t_n}^\text{reg}(m),\,y_{t_n}^\text{reg}(m)\right)_{m=1}^{M_\text{reg}^{\mathcal{C}_n}}$ be our new samples of regression pairs (where the indexation has been appropriately changed). Assuming $M_\text{reg}^{\mathcal{C}_n}\geq1$, we want to find a set of regression coefficients $\boldsymbol{\hat{\alpha}_{t_n}}=(\hat{\alpha}_{t_n}^{(1)},\ldots,\,\hat{\alpha}_{t_n}^{(B)})$ such that the following empirical $L_2$-risk is minimized \begin{equation}\label{L2riskOLS}
    \frac{1}{M_\text{reg}^{\mathcal{C}_n}}\sum_{m=1}^{M_\text{reg}^{\mathcal{C}_n}}\left|\sum_{b=1}^B\alpha_{t_n}^{(b)}\phi_b\left(x_{t_n}^\text{reg}(m)\right)-y_{t_n}^\text{reg}(m)\right|^2.
\end{equation}
For notational convenience, we introduce the compact notation $\boldsymbol{x_{t_n}}=\left(x_{t_n}^\text{reg}(1),\,\ldots,\,x_{t_n}^\text{reg}(M_\text{reg}^{\mathcal{C}_n})\right)$,\\ $\boldsymbol{y_{t_n}}=\left(y_{t_n}^\text{reg}(1),\,\ldots,\,y_{t_n}^\text{reg}(M_\text{reg}^{\mathcal{C}_n})\right)$ and \begin{equation*} \boldsymbol{\phi}(\boldsymbol{x_{t_n}})=\begin{pmatrix}\phi_1\left(x_{t_n}^\text{reg}(1)\right)&\phi_2\left(x_{t_n}^\text{reg}(1)\right)&\cdots&\phi_B\left(x_{t_n}^\text{reg}(1)\right)\\
\phi_1\left(x_{t_n}^\text{reg}(2)\right)&\phi_2\left(x_{t_n}^\text{reg}(2)\right)&\cdots&\phi_B\left(x_{t_n}^\text{reg}(2)\right)\\
\vdots & \vdots & \ddots & \vdots\\
\phi_1\left(x_{t_n}^\text{reg}(M_\text{reg}^{\mathcal{C}_n})\right)& \phi_2\left(x_{t_n}^\text{reg}(M_\text{reg}^{\mathcal{C}_n})\right)&\cdots&\phi_B\left(x_{t_n}^\text{reg}(M_\text{reg}^{\mathcal{C}_n})\right)
\end{pmatrix}.\end{equation*} It is a well-known fact (see \textit{e.g.}, \cite{regression_book}) that $\boldsymbol{\hat{\alpha}_{t_n}}$ is given by
\begin{equation}\label{alphaopt}
    \boldsymbol{\hat{\alpha}_{t_n}}=\left(\boldsymbol{\phi}(\boldsymbol{x_{t_n}})^{T}\boldsymbol{\phi}(\boldsymbol{x_{t_n}})\right)^{-1}\boldsymbol{\phi}(\boldsymbol{x_{t_n}})^{T}\boldsymbol{y_{t_n}},
\end{equation}
where we note that, if we choose linearly independent basis functions, matrix inversion in \eqref{alphaopt} exists almost surely since $X_{t_n}$ has a density\footnote{ In practice we run into troubles if we choose $B$ too high since the approximation of the matrix inverse may be unstable.}. We define the estimator
\begin{equation}\label{vOLS}
    \hat{v}_{B,K}(t_n,\cdot)=\sum_{b=0}^B\hat{\alpha}_{t_n}^{(b)}\phi_b(\cdot).
\end{equation}
If $M_\text{reg}^{\mathcal{C}_n}=0$ we know that all samples are in the exercise region and we simply set $\hat{v}_{B,K}(\cdot)\equiv g_{t_n}(\cdot)$. 
Since the LSM is one of the most established algorithms for valuation of Bermudan options, the theoretical properties are extensively studied and many of the results can also be applied to the algorithm above. However, we first need to make an assumption regarding $M_{\text{reg}}^{\mathcal{C}_n}$. Assume that there exists $c>0$ such that $\Q\{X_{t_n}\in\mathcal{C}_n\}\geq c$. It then holds for any $C\in\R$ that \begin{equation*}
    \Q\left\{\lim_{M_\text{reg}\to\infty}\sum_{m=1}^{M_\text{reg}}\I_{\{X_{t_n}^\text{reg}(m)\in\mathcal{C}_n\}}\geq C\right\} = 1,
\end{equation*}
which implies that $M_{\text{reg}}^{\mathcal{C}_n}$ approaches infinity when $M_\text{reg}$ approaches infinity almost surely.
Since the regression pairs are independently and identically distributed, it holds that $\hat{v}_{B,K}(t_n,X_{t_n})$ converges both in mean square and in probability to $v_{B}(t_n,X_{t_n})$ as $M_\text{reg}$ approaches infinity (see \textit{e.g.}, \cite{white}). To make it more clear when comparing the OLS-approximator of the option value to the neural network approximator (to be defined in the next section), we use the following notation \begin{equation}
    \hat{v}^{\text{OLS}}_{t_n}(\cdot) = \hat{v}_{B,M}(t_n,\cdot),
\end{equation}
where we assume that $B$ and $M$ are chosen such that both accuracy and time complexity are taken into account.

 A nice property of OLS regression is that, given $B$ and a sample of regression pairs, we have a closed-form expression for the optimal parameters and thus also the regression function \eqref{vOLS}. On the other hand, we may face memory or runtime issues for large $B$ and $M_\text{reg}^{\mathcal{C}_n}$ due to \eqref{alphaopt}. This is a problem, especially when we want to approximate a complicated function surface over a large approximation region. For example, consider an option based on 50 correlated assets. If we want to use the first and second-order polynomials as basis functions (including cross-terms) we have $B=\frac{50(50+3)}{2}=1325$, which is often too large for practical purposes. We should also have in mind that this corresponds to an approximation with polynomials of degree 2, which is usually not sufficient for complicated problems. There are however methods to get around this problem, see \textit{e.g.}, \cite{SGBM} in which the state space is divided into several bundles and regression is carried out locally at each bundle. Another suitable method to overcome these difficulties is neural network regression, which is presented in the next section.  

\subsection{Neural network regression}\label{NN_regress}
In this section we present, a simple neural network approximation of $V_{t_n}(\cdot)$. The neural network is a mapping, $v_{\varphi_n}\colon\R^d\to\R$, parametrized by the $p_{t_n}\in\N$ trainable parameters $\varphi_n\in\R^{p_{t_n}}$. The objective is to find $\varphi_n$ such that the empirical $L_2$-risk
\begin{equation}\label{L2riskNN}
    \frac{1}{M_\text{reg}^{\mathcal{C}_n}}\sum_{m=1}^{M_\text{reg}^{\mathcal{C}_n}}\left|v_{\varphi_n}(x_{t_n}^\text{reg}(m))-y_{t_n}^\text{reg}(m)\right|^2
\end{equation}
is minimized. We denote by $\hat{\varphi}_{n}$ an optimized version of $\varphi_n$ and define our neural network approximator of the option price at $t_n$ by \begin{equation*}
    \hat{v}^{\text{NN}}_{t_{n}}(\cdot) = v_{\hat{\varphi}_{n}}(\cdot).
\end{equation*}
To avoid  repetition, the description of the neural networks in Subsection \ref{NN_spec} is also valid for the neural network used here. However, one important difference is the output, which in this section is an approximation of the option value, and should hence take on values in $(0,\infty)$. A natural choice as activation function in the output layer is therefore $\text{ReLU}(\cdot)=\max\{0,\cdot\}$. Furthermore, by shifting the output with $-g_{t_{n}}(\cdot)$, \textit{i.e.,} designing the neural network to output $v_{\hat{\varphi}_n}(\cdot)-g_{t_n}(\cdot)$ and defining $\hat{v}^{\text{NN}}_{t_{n}}(\cdot) = v_{\hat{\varphi}_{n}}(\cdot)+g_{t_n}(\cdot)$ we can, for all $x\in\R^d$, enforce $\hat{v}^\text{NN}(x)\geq g_{t_n}(x)$ by using ReLU as activation function in the output layer. In many cases it seems to be beneficial to use the identity as the activation function in the output layer. This could possibly be explained by the fact that when using the ReLU as activation function, the gradient of the loss function, \eqref{L2riskNN} (with respect to the input) may vanish during training. This, in turn leads to an inefficient use of a gradient descent type algorithm in the optimization problem.

Another difference, which has to do with the training phase, is that the optimization of the parameters $\varphi_n$ does not have to be carried out recursively. This opens up the possibility for parallelization of the code.

By comparing \eqref{L2riskNN} to \eqref{L2riskOLS} we see that the optimization problems are similar. There are, however, some major differences. In Subsection \ref{OLSsec}, we have a closed-form expression for the optimal parameters resulting in the final regression function \eqref{vOLS}. This is not the case for the neural network regression and we therefore need to use an optimization algorithm to approximate the optimal parameters. On the other hand, as mentioned in Subsection \ref{OLSsec} it is sometimes hard to find basis functions that are flexible enough. This problem can be overcome with neural networks, which are known to be good global approximators.

\section{Approximation algorithms for exposure profiles}\label{EEapprox}
In this section, we introduce different ways to estimate \eqref{EE} and \eqref{PFE} relying on Monte-Carlo sampling and the approximation algorithms described in Sections \ref{learnSD} and \ref{learn_pw_optionvalues}. Furthermore, a simple example is presented and visualized, which aims to provide an intuitive understanding of the different methods. Finally, the advantages and disadvantages of each method are presented in a table.

In this section, the neural network-based approximation of the value function of the option, introduced in Subsection \ref{NN_regress} is used. However, it would have been possible to use the OLS-based approximation from Subsection \ref{OLSsec}, instead. 

We use $M\in\N$ independent realizations of $X$, which for $m\in\{1,2,\ldots,M\}$ are denoted by $x(m)=(x_t(m))_{n=0}^{N}$ for $t\in\mathbbm{T}$. When $X$ is given by \eqref{XPQtn}, realization $m$ is denoted by $\tilde{x}^{t_n}(m)=\big(\tilde{x}^{t_n}_t(m)\big)_{t\in[t_0,t_N]}$, where we recall that superscript $t_n$ refers to the point in time where the discontinuity of the drift coefficient is located. We introduce the following two approximations of the expected exposure under the risk neutral measure
\begin{align}\label{EEQ1}
    \hat{\text{EE}}^1_\Q(t_n)&=\frac{1}{M} \sum_{m=1}^{M}\hat{v}^{\text{NN}}_{t_n}(x_{t_n}(m))\I_{\left\{\tau\left(\boldsymbol{f}^{\hat{\boldsymbol{\theta}}}\right)>t_n\right\}},\\\label{EEQ2}
    \hat{\text{EE}}^2_\Q(t_n)&= \frac{1}{M}\sum_{m=1}^{M}\e^{-r\left(\tau_n\left(\boldsymbol{f}^{\boldsymbol{\hat{\theta}}}_n\right)-t_n\right)}g_{\tau_n\left(\boldsymbol{f}^{\boldsymbol{\hat{\theta}}}_n\right)}\left(x_{\tau_n\left(\boldsymbol{f}^{\boldsymbol{\hat{\theta}}}_n\right)}(m)\right)\I_{\left\{\tau\left(\boldsymbol{f}^{\boldsymbol{\hat{\theta}}}\right)> t_n\right\}}.
\end{align}
For the expected exposure under the real world measure, we have the following three approximations
\begin{align}\label{EEP1}
    \hat{\text{EE}}^1_\P(t_n)&=\frac{1}{M}\sum_{m=1}^{M}\hat{v}^{\text{NN}}_{t_n}\left(\tilde{x}^{t_n}_{t_n}(m)\right)\I_{\left\{\tau\left(\boldsymbol{f}^{\boldsymbol{\hat{\theta}}}\right)> t_n\right\}},\\\label{EEP2}
    \hat{\text{EE}}^2_\P(t_n)&= \frac{1}{M}\sum_{m=1}^{M}\e^{-r\left(\tau_n\left(\boldsymbol{f}^{\boldsymbol{\hat{\theta}}}_n\right)-t_n\right)}g_{\tau_n\left(\boldsymbol{f}^{\boldsymbol{\hat{\theta}}}_n\right)}\left(\tilde{x}^{t_n}_{\tau_n\left(\boldsymbol{f}^{\boldsymbol{\hat{\theta}}}_n\right)}(m)\right)\I_{\left\{\tau\left(\boldsymbol{f}^{\boldsymbol{\hat{\theta}}}\right)> t_n\right\}},\\\label{EEP3}
    \hat{\text{EE}}^3_\P(t_n)&=        \frac{1}{M}\sum_{m=1}^{M}\e^{-r\left(\tau_n\left(\boldsymbol{f}^{\boldsymbol{\hat{\theta}}}_n\right)-t_n\right)}g_{\tau_n\left(\boldsymbol{f}^{\boldsymbol{\hat{\theta}}}_n\right)}\left(x_{\tau_n\left(\boldsymbol{f}^{\boldsymbol{\hat{\theta}}}_n\right)}(m)\right)\I_{\left\{\tau\left(\boldsymbol{f}^{\boldsymbol{\hat{\theta}}}\right)> t_n\right\}}l\left(x_{t_n}(m),\,\ldots,\,x_{t_0}(m)\right),
\end{align}
where $l$ is the likelihood ratio function defined in Theorem \ref{EEtheorem}. We note that \eqref{EEQ1} and \eqref{EEP1} are the only approximations that require a calculation of the option value. On the other hand, we need $X$ to be described by a diffusion-type SDE in \eqref{EEP2} and we need to know the density functions to calculate \eqref{EEP3}.
To define the approximations of the potential future exposure, we start by defining the order statistic of $\left(\hat{v}^{\text{NN}}_{t_n}(x_{t_n}(m))\I_{\left\{\tau\left(\boldsymbol{f}^{\boldsymbol{\hat{\theta}}}\right)>t_n\right\}}\right)_{m=1}^{M}$, \textit{i.e.,} the vector given by $\left(\hat{v}^{\text{NN}}_{t_n}(x_{t_n}(\tilde{m}_1))\I_{\left\{\tau\left(\boldsymbol{f}^{\boldsymbol{\hat{\theta}}}\right)>t_n\right\}},\,\ldots,\,\hat{v}^{\text{NN}}_{t_n}(x_{t_n}(\tilde{m}_M))\I_{\left\{\tau\left(\boldsymbol{f}^{\boldsymbol{\hat{\theta}}}\right)>t_n\right\}}\right)$ satisfying $\hat{v}^{\text{NN}}_{t_n}(x_{t_n}(\tilde{m}_i))\leq \hat{v}^{\text{NN}}_{t_n}(x_{t_n}(\tilde{m}_j))$ whenever $i\leq j$. Furthermore, we define \begin{equation*}
    i_\alpha=\begin{cases}\lceil\alpha M\rceil,\quad\text{for }\alpha\geq 0.5,\\
    \lfloor\alpha M\rfloor,\quad\text{for }\alpha< 0.5.
    \end{cases}
\end{equation*}
The approximations of the potential future exposure are then defined as
\begin{align*}
    \hat{\text{PFE}}_\Q^\alpha(t_n)&=\hat{v}^{\text{NN}}_{t_n}(x_{t_n}(\tilde{m}_{i_\alpha}))\I_{\left\{\tau\left(\boldsymbol{f}^{\boldsymbol{\hat{\theta}}}\right)>t_n\right\}},\\
    \hat{\text{PFE}}_\P^\alpha(t_n)&=\hat{v}^{\text{NN}}_{t_n}(\tilde{x}^{t_n}_{t_n}(\tilde{m}_{i_\alpha}))\I_{\left\{\tau\left(\boldsymbol{f}^{\boldsymbol{\hat{\theta}}}\right)>t_n\right\}}.
\end{align*}
In Table \ref{specification}, some characteristics of the calculations needed for each approximation are given.
\begin{table}
    \centering
\begin{tabular}{ |p{1.75cm}||p{2cm}|p{2cm}|p{2cm}|p{2cm}|  }
 \hline
 \multicolumn{5}{|c|}{\textbf{Specification of requirements for each approximation}} \\
 \hline
 &  Approx- imation of the functional form $V_{t_n}(\cdot)$ &Sampling from $\tilde{X}^{t_n}$ & Known density functions& $X$ given must be by diffusion-type SDE \\
 \hline
 \multicolumn{1}{|c||}{$\hat{\text{EE}}_\Q^1$}   &     
 \multicolumn{1}{|c|}{\faCheck}    &\multicolumn{1}{|c|}{\faTimes}&   \multicolumn{1}{|c|}{\faTimes}&\multicolumn{1}{|c|}{\faTimes}\\
 \multicolumn{1}{|c||}{$\hat{\text{EE}}_\Q^2$} &   \multicolumn{1}{|c|}{\faTimes}  & \multicolumn{1}{|c|}{\faTimes}   &\multicolumn{1}{|c|}{\faTimes}& \multicolumn{1}{|c|}{\faTimes}\\
 \multicolumn{1}{|c||}{$\hat{\text{EE}}_\P^1$}  &\multicolumn{1}{|c|}{\faCheck}  & \multicolumn{1}{|c|}{\faCheck}&  \multicolumn{1}{|c|}{\faTimes}& \multicolumn{1}{|c|}{\faTimes}\\
 \multicolumn{1}{|c||}{$\hat{\text{EE}}_\P^2$}     &\multicolumn{1}{|c|}{\faTimes} & \multicolumn{1}{|c|}{\faCheck}&  \multicolumn{1}{|c|}{\faTimes}& \multicolumn{1}{|c|}{\faCheck}\\
 \multicolumn{1}{|c||}{$\hat{\text{EE}}_\P^3$} &   \multicolumn{1}{|c|}{\faTimes}  & \multicolumn{1}{|c|}{\faTimes}&\multicolumn{1}{|c|}{\faCheck}& \multicolumn{1}{|c|}{\faTimes}\\
\multicolumn{1}{|c||}{$\hat{\text{PFE}}_\Q$} & \multicolumn{1}{|c|}{\faCheck}   & \multicolumn{1}{|c|}{\faTimes}   &\multicolumn{1}{|c|}{\faTimes}& \multicolumn{1}{|c|}{\faTimes}\\
 \multicolumn{1}{|c||}{$\hat{\text{PFE}}_\P$}& \multicolumn{1}{|c|}{\faCheck}   & \multicolumn{1}{|c|}{\faCheck}&\multicolumn{1}{|c|}{\faTimes}& \multicolumn{1}{|c|}{\faTimes}\\
 \hline
\end{tabular}
    \caption{A specification of which approximations/entities that are required in order to carry out the different calculations. If required \faCheck, otherwise \faTimes.}
    \label{specification}
\end{table}

To explain the different approximations in a concrete setting we turn to a simple example.
\begin{example}
Consider a one-dimensional American put option, where we for simplicity assume that $r=0$. We are interested in the expected exposure and the potential future exposure at time $t_n\in(t_0,t_N)$ given that we have full knowledge of the market at time $t_0$. We give a short explanation of the intuition behind the different methods by referring to Figure \ref{MC_example}, where the problem is visualized. 

We start by $\hat{\mathrm{EE}}^1_\Q(t_n)$ and $\hat{\mathrm{EE}}^1_\P(t_n)$ for which we only use the figure to the left. For $\hat{\mathrm{EE}}^1_\Q(t_n)$ we follow the blue samples and note that samples 2 and 3 are not exercised prior to, or at $t_n$, which means that the indicator function in \eqref{EEQ1} equals 1. Sample 1, on the other hand, touches the exercise region prior to $t_n$ and has therefore already been exercised, which means that\begin{align*}
    \hat{\mathrm{EE}}^1_\Q(t_n)&=\frac{1}{3}\left(\hat{v}^{\mathrm{NN}}(x_{t_n}(1))\I_{\{\tau(\boldsymbol{f}^{\boldsymbol{\hat{\theta}}})>t_n\}}+\hat{v}^{\mathrm{NN}}(x_{t_n}(2))\I_{\{\tau(\boldsymbol{f}^{\boldsymbol{\hat{\theta}}})>t_n\}}+\hat{v}^{\mathrm{NN}}(x_{t_n}(3))\I_{\{\tau(\boldsymbol{f}^{\boldsymbol{\hat{\theta}}})>t_n\}}\right)\\&=\frac{1}{3}\left(\hat{v}^{\mathrm{NN}}(x_{t_n}(2))+\hat{v}^{\mathrm{NN}}(x_{t_n}(3))\right).\end{align*} When focusing on the red samples instead, we see that no sample touches the exercise region prior to $t_n$ and we obtain $$\hat{\mathrm{EE}}^1_\P(t_n)=\frac{1}{3}\left(\hat{v}^{\mathrm{NN}}(\tilde{x}^{t_n}_{t_n}(1))+\hat{v}^{\mathrm{NN}}(\tilde{x}^{t_n}_{t_n}(2))+\hat{v}^{\mathrm{NN}}(\tilde{x}^{t_n}_{t_n}(3))\right).$$
Similarly, we can \textit{e.g.,} state that $\hat{\mathrm{PFE}}_{\Q}^{2.5}=0$ and $\hat{\mathrm{PFE}}_{\P}^{97.5}=\hat{v}^{\mathrm{NN}}(\tilde{x}^{t_n}_{t_n}(2))$.

Moving on to $\hat{\mathrm{EE}}^2_\P(t_n)$ and $\hat{\mathrm{EE}}^3_\P(t_n)$, we shift focus to the figure to the right. For $\hat{\mathrm{EE}}^2_\P(t_n)$ we want to use the red samples and notice that samples 2 and 3 end up out of the money. We therefore obtain \begin{equation*}
    \hat{\mathrm{EE}}^2_\P(t_n) = \frac{1}{3}g_{\tau_n\left(\boldsymbol{f}^{\boldsymbol{\theta}}_n\right)}\left(\tilde{x}^{t_n}_{\tau_n\left(\boldsymbol{f}^{\boldsymbol{\theta}}_n\right)}(1)\right).
\end{equation*}
For $\hat{\mathrm{EE}}^3_\P(t_n)$, we instead consider the blue samples and see that sample 1 is exercised prior to $t_n$ and samples 2 and 3 end up in the money. However, we also need to adjust the estimate for using the wrong state process\footnote{The samples are generated from the state process under the $\Q-$measure between $t_0$ and $t_n$ which, if not corrected, would be in conflict with the definition of $\mathrm{EE}_\P(t_n)$.}. This is done by multiplying each term with the likelihood ratios $l(x(2))$ and $l(x(3))$ to finally obtain
\begin{equation*}
    \hat{\mathrm{EE}}^3_\P(t_n) = \frac{1}{3}\left(g_{\tau_n\left(\boldsymbol{f}^{\boldsymbol{\theta}}_n\right)}\left(x_{\tau_n\left(\boldsymbol{f}^{\boldsymbol{\theta}}_n\right)}(2)\right)l\left(x(2)\right)+g_{\tau_n\left(\boldsymbol{f}^{\boldsymbol{\theta}}_n\right)}\left(x_{\tau_n\left(\boldsymbol{f}^{\boldsymbol{\theta}}_n\right)}(3)\right)l\left(x(3)\right)\right).
\end{equation*}
The last estimate, $\hat{\mathrm{EE}}^2_\Q(t_n)$, is obtained by removing the likelihood ratios from the estimate for $\hat{\mathrm{EE}}^3_\P(t_n)$.
 
 \begin{figure}[htp]
\centering
\begin{tabular}{ccc}
     \includegraphics[width=82mm]{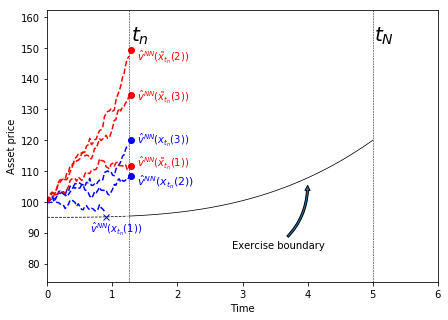}&
          \includegraphics[width=82mm]{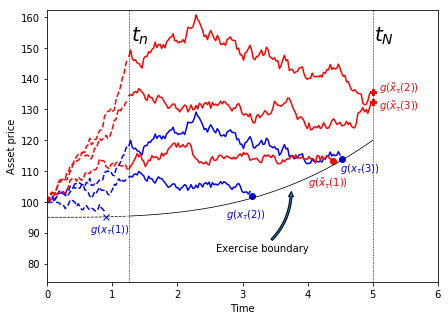}\\
\end{tabular}
\caption{Blue trajectories are distributed as $X$ and red trajectories are distributed as $\tilde{X}^{t_n}$. The boundary for immediate exercise is pointed out and should be interpreted as, as soon a trajectory touches the boundary, the option is exercised and the holder receives the immediate pay-off. Recall that the exercise boundaries are calculated in order to be optimal under the $\Q$-measure.}
\label{MC_example}
\end{figure}
\end{example}
To conclude, we note that Figure \ref{MC_example} displays, to the left, the cases where functional form approximations of the option values are used and to the right the cases where cashflow-paths are used (this can be read in Table \ref{specification}). Furthermore, blue and red trajectories are distributed according to $X$ and $\tilde{X}^{t_n}$, respectively (this can also be read in Table \ref{specification}).

\section{Numerical results}\label{num_res}
This section is divided into two parts: in the first part we use the Black--Scholes dynamics for the underlying assets. The proposed algorithm is compared with two different Monte-Carlo-based algorithms for a two-dimensional case. We focus on both the accuracy of the computed exercise boundaries and the exposure profiles. Furthermore, the exercise boundary is compared with the exercise boundary for the corresponding American option, computed by a finite element method, from the PDE-formulation of the problem. Exposure profiles are then computed both under the risk neutral and the real world measures for problems in 2 and 30 dimensions. Finally, we compare the OLS-regression with the NN-regression. 

In the second part we consider stochastic volatility and compute exposure profiles under the Heston model.
\subsection{Black--Scholes dynamics}
In the Black--Scholes setting the only risk factors are the $d\in\N$, underlying assets, denoted by $S$. We assume a constant risk-free interest rate $r\in\R$, and for each asset $i\in\{1,2,\ldots,d\}$, volatility $\sigma_i\in(0,\infty)$, and constant dividend rate $q_i\in(0,\infty)$. The state process $X$ is then given by the asset process $S=(S_t)_{t\in[t_0,t_N]}$, \textit{i.e.,} $X=S$, following the SDE 
\begin{equation}\label{XA}
    \frac{\d (S_t)_i}{(S_t)_i} = (A_i-q_i)\d t + \sigma_i\d (W_t^\A)_i,\ (S_{t_0})_i=(s_{t_0})_i;\quad t\in[t_0,\,t_N], 
\end{equation}
with initial state $(s_{t_0})_i\in(0,\infty)$, and where $\A$ is either the real world measure $\P$ or the risk neutral measure $\Q$. In the real world setting $A_i=\mu_i\in\R$ and in the risk neutral setting $A_i=r$. The process $W^\A=(W_t^\A)_{t\in[t_0,t_n]}$ is a $d$-dimensional Brownian motion under the measure $\A$. Furthermore, $W^\A$ is correlated with correlation parameters $\rho_{ij}\in[-1,1]$, \textit{i.e.,} for $i,j\in\{1,2,\ldots,d\}$, we have that $\E^\A[\d(W_t^\A)_i \d(W_t^\A)_j]=\rho_{ij}\d t$, with $\rho_{ii}=1$. Moreover, for $t_0\leq u\leq t\leq t_N$ a closed-form solution to \eqref{XA} is given by \begin{equation}
    (S_t)_i=(S_u)_i\text{exp}\left((A_i-q_i-\frac{1}{2}\sigma_{i}^2)(t-u)+ \sigma_{i}\left((W^\A_t)_i-(W^\A_u)_i\right)\right).
\end{equation}
We note that $\log{S}$ has Gaussian increments under both the $\P-$ and the $\Q-$measures which implies that we have access to a closed-form expressions for the transition densities of $S$, and in turn also to the likelihood ratio in \eqref{EEP3}. 
\subsubsection{Bermudan max-call option}\label{Berm_max_call}
 At exercise, a Bermudan max-call option pays the positive part of the maximum of the underlying assets after subtraction of a fixed amount $K\in\R$. This implies an identical pay-off function at all exercise dates, given by \begin{equation*}
    g(s)=\left(\max\left\{s_1,\,s_2\,\ldots,\,s_d\right\}-K\right)^+,
\end{equation*} 
where $s=(s_1,s_2,\ldots,s_d)\in (0,\infty)^d$ and for $c\in\R$, $(c)^+=\max\{c,\,0\}$. 

We choose to focus on Bermudan max-call options for two reasons: first, there exist plenty of comparable results in the literature (see \textit{e.g.,}  \cite{DOS}, \cite{maxcall}, \cite{SGBM}); second, and more importantly, the exercise region is nontrivial with several interesting features. For example, $\max\{s_1,s_2,\ldots,s_d\}$ is not a sufficient statistic for making optimal exercise decisions, meaning that we cannot easily reduce the dimensionality (all dimensions are needed in order to price the option). This is not the case with, \textit{e.g.,} the geometric-average call option with pay-off function $g(s)=\left(\frac{1}{d}\prod_{i=1}^d s_i-K\right)^+$, since $\frac{1}{d}\prod_{i=1}^d s_i$ is a sufficient statistic for optimal exercise decisions (only the geometric average is needed to price the option, meaning that the problem can be reduced to 1 dimension, see \textit{e.g.,}  \cite{MonteCarloFinance}). Another example is the arithmetic-average call option with pay-off function $g(s)=\left(\frac{1}{d}\sum_{i=1}^d s_i-K\right)^+$. Similar to the max-call option, $\frac{1}{d}\sum_{i=1}^d s_i$ is not a sufficient statistic for optimal exercise decisions, but on the other hand the exercise region is convex\footnote{In this section we have assumed that results for exercise regions for American options also hold for their Bermudan counterparts.}$^{,}$\footnote{Convex in the underlying assets for fixed $t$.} (see \cite[Proposition 6.1]{ExRegBroadie}). Convexity of the exercise region does not hold for the max-call option, making it hard to capture the exercise region when global polynomials are used as basis functions in \textit{e.g.,} the LSM. Methods which instead rely on local regression can, to some extent, overcome this problem but it is difficult to decide how the localization (usually localization of the state space) should be done, especially in high dimensions.

In the numerical experiments we use the following parameters $r=0.05$ and for $i,j\in\{1,2,\ldots,d\}$, $q_i=0.1$, $\sigma_{i}=0.2$, for $i\neq j$, $\rho_{ij}=0$, $N=9$, $t_0=0$, $t_N=3$, $(s_{t_0})_i=100$ and $K=100$. We want to emphasize that the choice of having no correlation between assets is due to the fact that this case has been studied thoroughly in the literature (see \textit{e.g.,} \cite{DOS}, \cite{maxcall}, \cite{SGBM}, \cite{ExRegBroadie}). To verify that the algorithm is also able to tackle problems with correlated assets, we replicated an experiment in \cite{SGBMNN}, of pricing a put option on the arithmetic average of 5 correlated assets and obtained the same price\footnote{We obtained the price 0.1804, which is the same price up to the given accuracy of 4 digits, presented in \cite[Parameters as in Set II, results in Table 3 on p. 22]{SGBMNN}.}. 

\subsubsection{Approximation of the option value at initial time}
The performance of the DOS algorithm is thoroughly explored for a wide range of different examples in \cite{DOS}. However, the convergence with respect to the amount of training data used, $M_{\text{train}}$ was not given. We therefore present, in Figure \ref{conv_NN}, an example of how the option price at $t_0$ is converging to a reference value in terms of the amount of training data. In the considered example, we use the parameter values given at the end of Subsection \ref{Berm_max_call} with $d=2$, \textit{i.e.,} a two-dimensional Bermudan max-call option. The reference value (13.902) for $V_{t_0}(x_{t_0})$ is computed by a binomial lattice model in \cite{Binomial_price}. The DOS approximation, denoted by $\hat{V}_{t_0}(x_{t_0})$, is computed according to \eqref{V0_est}. To be more specific, one neural network is trained for each $M_{\text{train}}\in\{2^{12}, 2^{14},2^{16},2^{18},2^{20},2^{22},2^{25}\}$, as described in Subsection \ref{train_NN}. For each $M_{\text{train}}$, the option value, $V_{t_0}(x_{t_0})$ is computed 20 times (with 20 independent realizations of $X$) with $M_\text{val}=2^{20}$ and the average of these 20 values is showed in Figure \ref{conv_NN}. Furthermore, the figure to the right displays empirical $95\%$-confidence intervals of the sample mean, which is computed by adding and subtracting $\frac{1.96}{\sqrt{19}}$ times the sample standard deviation.

 \begin{figure}[htp]
\centering
\begin{tabular}{ccc}
     \includegraphics[width=82mm]{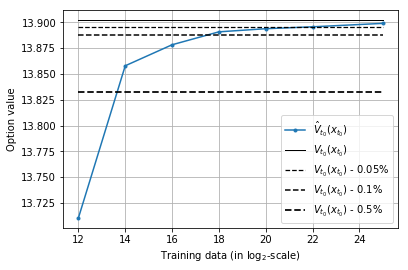}&
          \includegraphics[width=82mm]{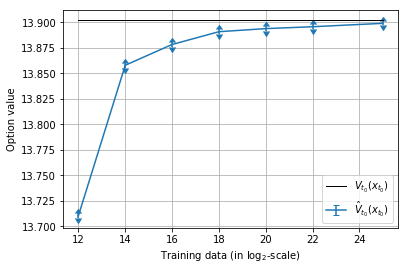}\\
\end{tabular}
\caption{Convergence of approximate option values in the amount of training data. Reference value 13.902, computed by a binomial lattice model in \cite{Binomial_price}. \textbf{Left:} With different levels of deviation from the reference value. \textbf{Right:} With empirical $95\%$-confidence intervals of sample mean.}
\label{conv_NN}
\end{figure}

\subsubsection{Comparison with Monte-Carlo-based algorithms}\label{comp_MC_algo}
We start with a short introduction of the two Monte-Carlo-based algorithms with which we compare results for the two-dimensional max-call option. 

The least squares method (LSM) proposed in \cite{LSM}, is one of the most used methods for pricing American/Bermudan options. The method approximates exercise boundaries and computes the option value from discounted cashflows. The regression is carried out globally, \textit{i.e.,} with the same regression function over the entire state space. However, if one is only interested in the option value at $t_0$, it is beneficial to only consider ITM-samples in the regression state. This is done when the exercise boundary, and $\text{PFE}_{97.5}$ are approximated. For approximating $\text{EE}$ and $\text{PFE}_{2.5}$ we need the entire distribution of future option values forcing us to also include OTM-samples in the regression. We use as basis functions the first 6 Laguerre polynomials $L_n(x)=\e^{-x/2}\frac{\e^x}{n!}\frac{\d^n}{\d x^n}(x^n\e^{-x})$ of $(S_{t_n})_1$ and $(S_{t_n})_2$ and the first 4 powers of $\max\{\log((S_{t_n})_1),\log((S_{t_n})_2)\}$ for all $t_n\in\mathbbm{T}$. Note that since we have no correlation between $(S)_1$ and $(S)_2$, we do not include cross-terms in the basis.

The second algorithm is the stochastic grid bundling method (SGBM) proposed in \cite{SGBM}, in which regression is carried out locally in so-called bundles. In the SGBM the target function in the regression phase is the option value which makes it suitable for approximations of exposure profiles. There are several papers on computations of exposure profiles with the SGBM, see \textit{e.g.,} \cite{SGBM_Heston}.  At each exercise date $t_n\in\mathbbm{T}$, we use as basis functions: a constant, the first 4 powers of $(S_{t_n})_1$ and $(S_{t_n})_2$, and the first 2 powers of $\max\{\log((S_{t_n})_1),\log((S_{t_n})_2)\}$. Furthermore, the state space is divided into 32 equally-sized bundles based on $\max\{(S_{t_n})_1,(S_{t_n})_2\}$. 

Before presenting any results, we recall from equations \eqref{EE} and \eqref{PFE} that the expected exposure and the potential future exposure are two statistics of \begin{equation}V_{t_n}(S_{t_n})\I_{\{\tau>t_n\}},\label{IV}\end{equation}
where $\tau$ is an $S-$stopping time. This means that approximations of the exposure profiles are sensitive to, not only the option value itself, but also the exercise strategy. The SGBM and LSM only compute the exercise strategy implicitly, \textit{i.e.,} by comparing the approximate option value and the immediate pay-off. Therefore small errors of the approximate option value close to the exercise boundary can lead to significant errors in the exercise strategy\footnote{ In our experiments the errors of the approximate option values seem be of the same sign locally \textit{i.e.,} the polynomial basis function underestimates the option value in some regions and overestimates the option value in other regions.}. We therefore start by presenting a comparison of the exercise boundaries. As a reference, we use the exercise boundary for the corresponding American option, which is computed from the PDE-formulation with the finite element method used in \cite{FEM_ex_reg}. We note that since the PDE formulation refers to the American option, the exercise boundary differs slightly\footnote{In fact, the continuation region for an American option is a subset of the continuation region for the Bermudan counterpart.} from the exercise boundary of the Bermudan counterpart, which we are interested in. 

In Figure \ref{test_N9}, a comparison of the exercise boundaries at $t_8\approx 2.67$ for the different algorithms is presented. As we can see, the DOS algorithm captures the shape of the exercise regions while both the SGBM and the LSM seem to struggle, especially with the part of the continuation region along (and around) the line $(S_{t_n})_1=(S_{t_n})_2$, in particular for high values of $(S_{t_n})_1$ and $(S_{t_n})_2$. The irregular shaped features in the continuation region for the SGBM are a consequence of the local regression. In particular, the triangular shapes come from the specific choice of bundling rule (bundling based on $\max\{(S_{t_n})_1,(S_{t_n})_2\}$).

Moving on to the exposure profiles, we see in Figure \ref{diff_algo}, that even though the DOS algorithm and the SGBM seem to agree on the exposure profile in general, we notice a difference in the $\text{PFE}_{97.5}$. This is a consequence of the slight bias towards classifying samples as belonging to the continuation region, which is shown in Figure \ref{diff_algo}, to the right. 

The LSM is performing worse, both in terms of accuracy of exposure profiles and bias towards miss-classification. This is however not a surprise, since the LSM is tailored to calculate the option value at $t_0$.

Finally, it should be pointed out that for both the SGBM and LSM, it could very well be the case that another set of basis functions would better capture the shape of the exercise boundaries. In this two-dimensional example one could probably use geometric intuition to come up with a better set of basis functions, but in higher dimensions, and for more complicated pay-off functions, this becomes difficult.

\begin{figure}[htp]
\centering
\begin{tabular}{ccc}
     \includegraphics[width=82mm]{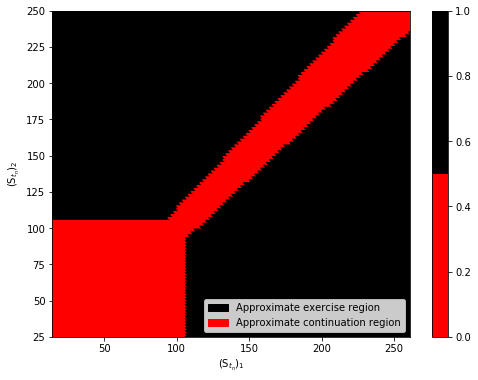}  & \includegraphics[width=82mm]{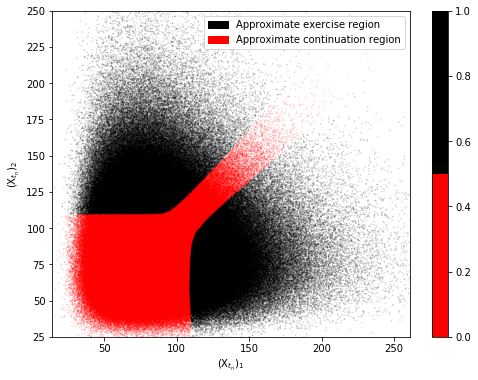}\\
          \includegraphics[width=82mm]{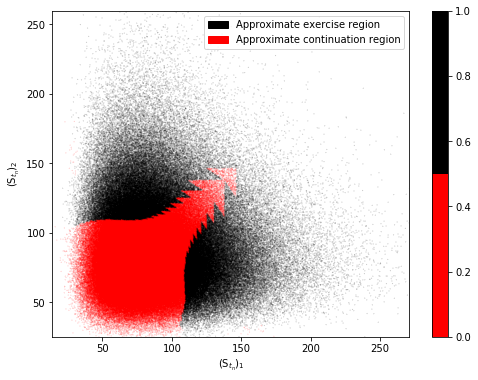} &  \includegraphics[width=82mm]{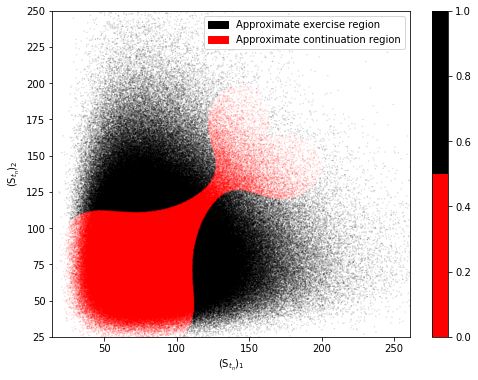}\\
\end{tabular}
\caption{Approximate exercise boundaries for a two-dimensional max-call option at $t_8\approx2.67$. \textbf{From top left to bottom right:} FEM (American option), DOS, SGBM and LSM respectively.}
\label{test_N9}
\end{figure}

 \begin{figure}[htp]
\centering
\begin{tabular}{ccc}
     \includegraphics[width=82mm]{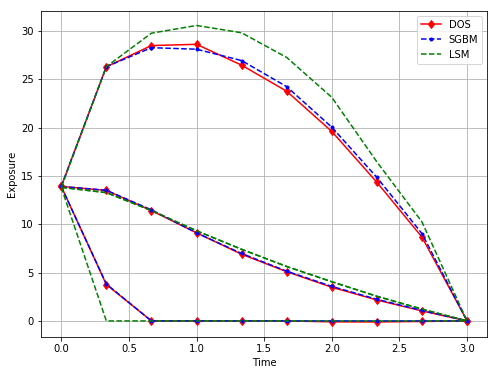}&
          \includegraphics[width=82mm]{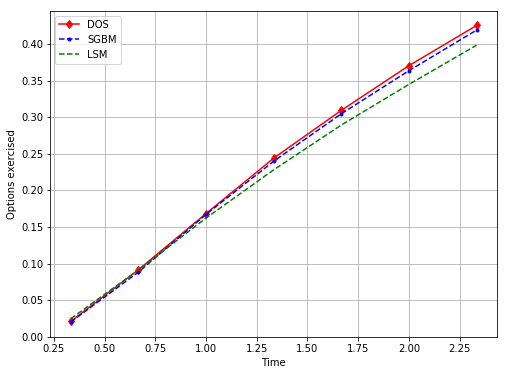}\\
\end{tabular}
\caption{Comparison of DOS, SGBM and LSM for a two-dimensional max-call option. \textbf{Left:} Expected exposure and potential future exposures at 97.5\%- and 2.5\%-levels. \textbf{Right:} Proportion of options exercised at different exercise dates.}
\label{diff_algo}
\end{figure}

\subsubsection{Exposure profiles under different measures}
In this section we compare exposure profiles under different measures for the max-call option, in 2 and 30 dimensions. In \textit{Case I} we set $d=2$ and $\P^1$ and $\P^2$ such that for $i\in\{1,2\}$, we have drifts $(\mu_1)_i=15\%$ and $(\mu_2)_i=-5\%$. In \textit{Case II} we set $d=30$ and $\P^1$ and $\P^2$ such that such that for $i\in\{1,\ldots,30\}$, we have drifts $(\mu_1)_i=7.5\%$ and $(\mu_2)_i=2.5\%$.

Figure \ref{Exposure_prof_2} shows exposure profiles in 2 and 30 dimensions on the left side. On the right, we see a comparison of the different ways to compute expected exposures which all agree to high accuracy. Furthermore, Figure \ref{Ex_freq_test}, displays that the fraction of exercised options over time is highly dependent on the choice of measure.

\begin{figure}[htp]
\centering
\begin{tabular}{ccc}
     \includegraphics[width=82mm]{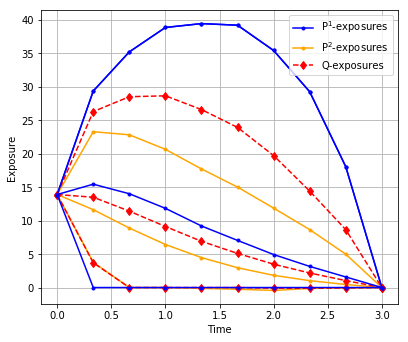}  & \includegraphics[width=82mm]{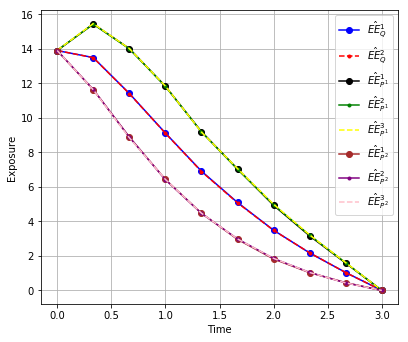}\\
          \includegraphics[width=82mm]{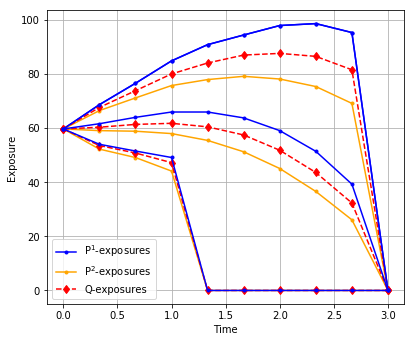} & \includegraphics[width=82mm]{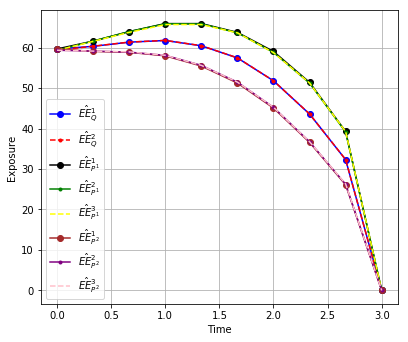}\\
\end{tabular}
\caption{Approximate exposures, at the exercise dates over the life time of the contract. \textbf{Top:} Case I. \textbf{Bottom:} Case II. \textbf{Left:} For $i=1$ and $i=2$, $\hat{\text{PFE}}^{97.5}_{\P^i}$, $\hat{\text{PFE}}^{97.5}_\Q$, $\hat{\text{EE}}_{\P^i}^1$, $\hat{\text{EE}}_\Q^1$  $\hat{\text{PFE}}^{2.5}_{\P^i}$ and $\hat{\text{PFE}}^{2.5}_\Q$. \textbf{Right:} For $i=1$ and $i=2$, $\hat{\text{EE}}_\Q^1$, $\hat{\text{EE}}_\Q^2$, $\hat{\text{EE}}_{\P^i}^1$, $\hat{\text{EE}}_{\P^i}^2$ and $\hat{\text{EE}}_{\P^i}^3$.}
  \label{Exposure_prof_2}
  \end{figure} 

\begin{figure}[htp]
\centering
\begin{tabular}{ccc}
     \includegraphics[width=91mm]{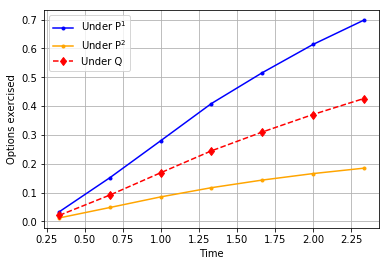} 
\end{tabular}
\caption{Proportion of samples exercised at different exercise dates under measures $\Q,\ \P^1$ and $\P^2$ for Case I.}
\label{Ex_freq_test}
\end{figure}

\subsubsection{Comparison of the OLS-regression and the NN-regression for approximation of pathwise option values}

Finally, we compare the performance of the OLS-regression, with the NN-regression, introduced in Subsections \ref{OLSsec} and \ref{NN_regress}. We emphasise that both OLS-regression and NN-regression are regressing the current state on discounted cashflow paths approximated by the DOS algorithm (described in Section \ref{learnSD}). They can therefore be seen as phase 2 of an algorithm producing pathwise option values. 

After conducting numerical experiments on a variety of different examples we conclude that the expected exposures are very similar for the two regression methods. However, the potential future exposure is not always captured by the OLS-regression. The difficulty lies in finding a set of computationally feasible basis functions, flexible enough to accurately capture a complicated function surface on a large domain (similar problem as for the LSM in Subsection \ref{comp_MC_algo}). To overcome this problem, we also implement a slightly different version of the algorithm, where we instead carry out local regression in the continuation regions. To be able to differentiate between the local and global OLS-regressions, we denote the regression functions by $v^{\text{OLS}_{\text{loc}}}$ and  $v^{\text{OLS}_{\text{glob}}}$, respectively. The localization procedure is done similarly as in the SGBM, \textit{i.e.,} at each $t\in\mathbbm{T}$, the state space is divided into bundles of equal size (in terms of the number of samples in each bundle), based on $\max\{(S_t)_1,(S_{t})_2\}$. With local OLS-regression we obtain almost identical exposure profiles as with NN-regression. In Figure \ref{OLS_NN}, on top to the left, the exposure profiles computed with the three different algorithms are displayed. Furthermore, from top right to bottom right, we compare the approximate risk premia for holding instead of immediately exercising the option at $t_2\approx0.667$ and some $x\in\R^2$, \textit{i.e.,} $v_{t_2}^{Z}(x)-g_{t_2}(x)$, with $Z$ representing, in order $\text{NN}$, $\text{OLS}_{\text{glob}}$ and $\text{OLS}_{\text{loc}}$. We know that for all $x\in\R^2$, it holds that $V_{t_2}^{Z}(x)-g_{t_2}(x)\geq0$. We see in Figure \ref{OLS_NN}, top right, that this is captured by the NN-regression, since the values range from 0 to just above 12. When we carefully evaluate the values of the risk premia computed with local and global OLS-regression (Figure \ref{OLS_NN} bottom left and bottom right) we see that negative values exist in both plots. We see similar phenomena for high values, \textit{i.e.,} the range is stretched upwards in comparison to the values obtained with NN-regression. The reason for the negative values is that $v^{\text{OLS}_\text{loc}}$ and $v^{\text{OLS}_\text{glob}}$ underestimate the option values close to the boundary (which in this case coincides with the exercise boundary since the regression is carried out only in the continuation region). To compensate for this, we see a tendency of higher values in the center (not close to the exercise boundaries) of the continuation regions. As a final remark, we note that this behaviour is reduced by using local regression instead of global regression (the range of values in Figure \ref{OLS_NN} is tighter for local than global regression). On the other hand, we note discontinuities in Figure \ref{OLS_NN} bottom, left, stemming from the localized regression, of the risk premium computed with the local OLS-regression. The discontinuity in Figure \ref{OLS_NN}, bottom right, is a boundary issue of the OLS-regression (regression is only carried out in the continuation region since we set the value equal to the immediate pay-off in the exercise region). 

Even though, the local OLS-regression is less accurate when it comes to computations of exposure profiles than the other two algorithms in this section, it is clear by comparing Figures \ref{diff_algo} and \ref{OLS_NN}, that it outperforms the LSM. This is not a surprise since:
\begin{enumerate}
    \item The accumulated error in the LSM (because of recursive dependency of the regression functions) is significantly reduced since the regression functions are not sequentially dependent. The discounted cashflows are projected onto the risk factors directly. This implies that the only error accumulation (over time), in the OLS-regression originates from the DOS algorithm, which computes the exercise boundaries with high accuracy;
    \item By recalling equation \eqref{IV}, we note that a less accurate stopping strategy may cause a less accurate exposure.  
\end{enumerate}

\begin{figure}[htp]
\centering
\begin{tabular}{ccc}
     \includegraphics[width=82mm]{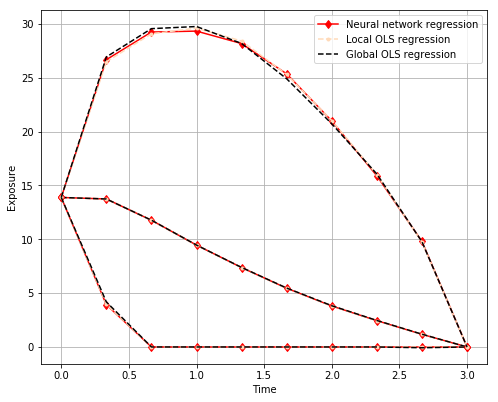}  & \includegraphics[width=82mm]{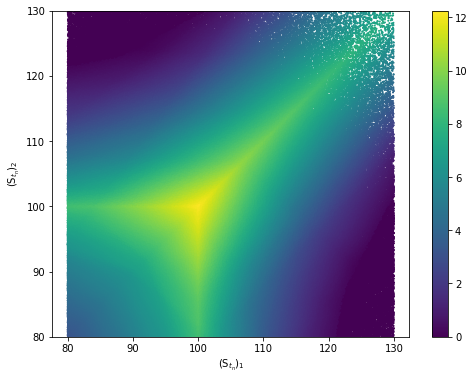}\\
          \includegraphics[width=82mm]{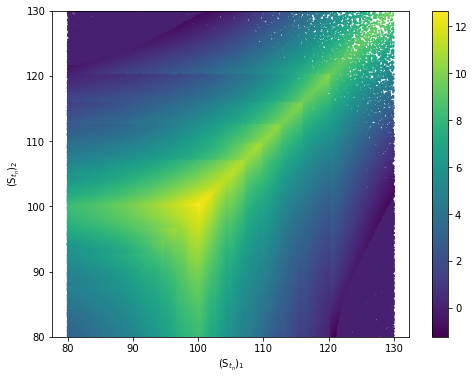}  & 
          \includegraphics[width=82mm]{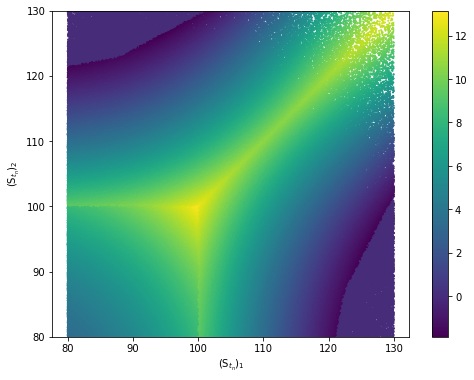}\\
\end{tabular}
\caption{Comparison of NN-regression, local OLS-regression and global OLS-regression for a two-dimensional max-call option. For all three regression techniques, the DOS is used to approximate the optimal stopping strategy. \textbf{Top left:} Exposure profiles. \textbf{From top right to bottom left:} Approximate risk premium for holding option instead of immediate exercise at $t_2\approx0.667$, \textit{i.e.,} $v_{t_2}^{Z}(\cdot)-g_{t_2}(\cdot)$, with $Z$ representing $\text{NN}$, $\text{OLS}_{\text{loc}}$ and $\text{OLS}_{\text{glob}}$, respectively.}
\label{OLS_NN}
\end{figure}

\subsection{Heston model dynamics}
In this section we assume a one-dimensional underlying asset following the Heston stochastic volatility model \cite{Heston}, which is considered only under the risk neutral measure. We therefore omit the explicit notation of the probability measure used in this section.
In this setting, the market is described by, not only the underlying asset price process $S=(S_t)_{t\in[t_0,t_N]}$ itself , but also by the instantaneous variance process $\nu=(\nu_t)_{t\in[t_0,t_N]}$. The state process is then the two-dimensional process $X=(\nu,S)$ which satisfies the system of SDEs  
\begin{align}\label{S_Heston}
    \d S_t &= (r-q)S_t\d t + \sqrt{\nu_t}S_t\d W_t^{S},\ S_{t_0}=s_{t_0};\quad t\in[t_0,\,t_N],\\
    \d\nu_t&=\kappa(\theta-\nu_t)\d t + \xi\sqrt{\nu_t}\d W_t^{\nu},\ \nu_{t_0}=\nu_0;\quad t\in[t_0,t_N],\label{v_Heston}
\end{align}
with risk-free interest rate $r\in\R$, dividend rate $q\in(0,\infty)$, initial conditions $s_{t_0},\nu_0\in(0,\infty)$, speed of mean reversion $\kappa\in(0,\infty)$, long term mean of the variance process $\theta\in(0,\infty)$, and volatility coefficient of the variance process $\xi\in (0,\infty)$. Furthermore, $(W_t^S)_{t\in[t_0,t_N]}$ and $(W_t^{\nu})_{t\in[t_0,t_N]}$ are two one-dimensional, standard Brownian motions satisfying $\E[\d W_t^S\d W_t^\nu]=\rho_{\nu,S}\d t$ for some correlation parameter $\rho_{\nu,S}\in(-1,1)$. We, however notice that it is important to be careful when using the Heston model, since for some parameters, moments of higher order than 1 can become infinite in finite time (see \cite[Proposition 3.1]{Pieterbarg}). Equation \eqref{v_Heston} is the SDE for the well-established Cox--Ingersoll--Ross (CIR) process, introduced in \cite{CIR}. When the so-called Feller condition \begin{equation*}
    2\kappa\theta\geq \xi^2, 
\end{equation*}
is satisfied, it holds that 0 is an unattaiable boundary for $\nu$. If the Feller condition is not satisfied, then 0 is an attainable, but strongly reflective\footnote{ Strongly reflective in the sense that the time spent at 0 is of Lebesgue measure zero, see \textit{e.g.,} \cite{Pieterbarg}.} boundary, see \textit{e.g.,} \cite{Feller}. This leads to an accumulation of probability mass around zero, which makes it more challenging to approximate $\nu$ accurately. Unfortunately, the Feller condition is rarely satisfied for parameters calibrated to the market.  

In this paper, we use the QE-scheme, proposed in \cite{QE_Heston} is used to approximate\footnote{For notational convenience, the state process is denoted by $X$, which falsely indicates that we have an exact form for $X$. This is because our focus is on approximating option values, not the underlying state process. It should however be mentioned that $X$ needs to be approximated in this section.} $(\nu,S)$. If necessary, we choose a finer time grid for the approximation of $(S,\nu)$ than the exercise dates, $\mathbbm{T}$. 

We consider a standard Bermudan put option, \textit{i.e.,} identical pay-off functions at all exercise dates, only depending on the underlying asset. Furthermore, the pay-off function for the Bermudan put option is given by \begin{equation*}
    g(s)=(K-s)^+,
\end{equation*} 
for $s\in(0,\infty)$, and strike $K\in\R$.

\subsubsection{Comparison with Monte-Carlo-based algorithms}
In this section, we again compare the DOS algorithm with the two Monte-Carlo-based algorithms, SGBM and LSM. We use the following set of model parameters: $r = 0.04$, $q=0$, $s_{t_0}=100$, $\kappa=1.15$, $\theta = 0.0348$, $\xi = 0.459$, $\nu_0 = 0.0348$ and $\rho_{\nu,S} = -0.64$, and the contract parameters: $T=0.25$, $N=10$, $K=100$. The parameters coincide with Set B in \cite{Fourier_Ruijter}, in which the valuation is carried out with the so-called 2D-COS method. The 2D-COS method is a Fourier-based method and is assumed to yield highly accurate valuation of the option. 

For the LSM, we use as basis functions, for $t_n\in\mathbbm{T}$, Laguerre polynomials of degree 3 of $S_{t_n}$, Laguerre polynomials of degree 3 of  $\nu_{t_n}$ and $\nu_{t_n}S_{t_n}$ (only on constant basis function is used). For the SGBM, we use 32 equally-sized bundles based on $S_{t_n}$ and for $t_n\in\mathbbm{T}$, we use as basis functions a constant, $S_{t_n}$, $S_{t_n}\nu_{t_n}$ and the first 3 powers of $\nu_{t_n}$. These parameters are chosen such that the approximate option value at $t_0$ are as close as possible to the (almost) exact value of 3.198944, for a Bermudan option with 10 exercise dates, retrieved from \cite{Fourier_Ruijter}. The obtained option values (for each algorithm, average value of 10 runs) at $t_0$ were 3.1792, 3.2033, and 3.1222 for the DOS algorithm, the SGBM and the LSM, respectively. 

It should be stressed that the numerical results for the LSM and the SGBM in this section should not be seen as state of the art performance of the algorithms. For example, in \cite{SGBM_Heston}, a bundling scheme based on recursive bifurcation and rotation of the state space to match the correlation between $S$ and $\nu$ gave accurate results. Furthermore, they use as basis functions only monomials of $\log{(S_{t_n})}$ and letting the stochastic variance $\nu_{t_n}$ enter only through conditional expectations of the form $\E[\left(\log(S_{t_n})\right)^k\,|\,S_{t_{n-1}},\nu_{t_{n-1}}]$. These conditional expectations are computed from the characteristic function of the $S_{t_n}\,|\,S_{t_{n-1}},\nu_{t_{n-1}}$ which was presented in \cite{Heston}. Similar improvements could also be done for the LSM. The reason for this comparison to still be relevant is to demonstrate the flexibility of DOS and the NN-regression, in which nothing has been changed (from the examples with Black--Scholes dynamics) except the dynamics of the stochastic process from which we generate training data.

In Figure \ref{ex_reg_Heston}, the exercise boundaries are presented in the state space. Worth noticing is that for $t_n\in\mathbbm{T}$, both the option value and the pay-off functions are increasing in $\nu_{t_n}$ and decreasing in $S_{t_n}$. An immediate consequence of this is that we can have at most one exercise boundary along the lines with constant $S_{t_n}$ or constant $\nu_{t_n}$. We can therefore conclude inconsistencies in the exercise boundaries for both the LSM and the SGBM. In Figure \ref{ex_reg_Heston}, on the bottom line to the right, the empirical probability density functions (pdf) of the exposures at the exercise dates are plotted.
\begin{figure}[htp]
\centering
\begin{tabular}{ccc}
     \includegraphics[width=82mm]{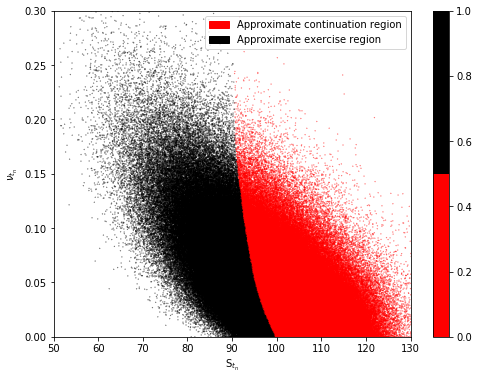}  & \includegraphics[width=82mm]{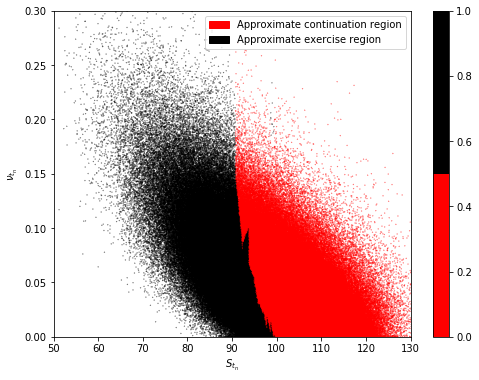}\\
          \includegraphics[width=82mm]{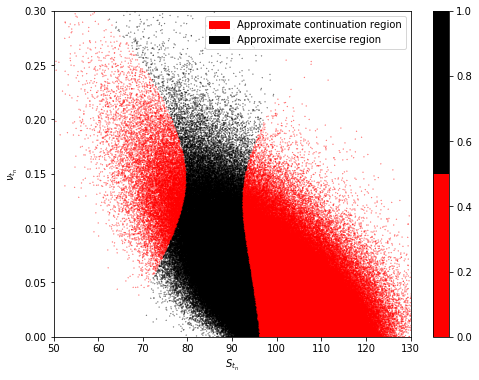}  & 
          \includegraphics[width=82mm]{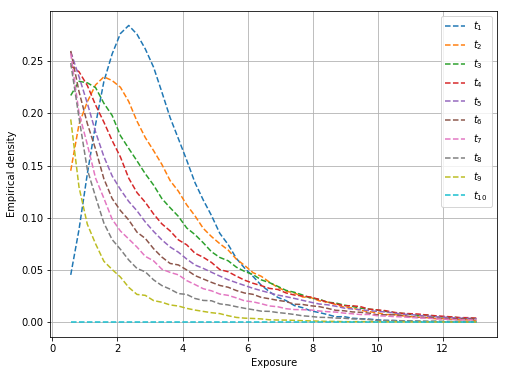}\\
\end{tabular}
\caption{Approximate exercise boundaries for a Bermudan put option under the Heston model at $t_9=0.225$ and the empirical density for the exposure at all exercise dates. \textbf{From top left to bottom right:} DOS, SGBM, LSM and the empirical density of the exposure.}
\label{ex_reg_Heston}
\end{figure}
Figure \ref{diff_algo_Heston} shows the exposure profiles and the exercise frequency computed with the three algorithms. We note that the DOS and the SGBM seem to agree fairly well on the EE and the lower percentile of the PFE but differ significantly on the upper PFE. 
\begin{figure}[htp]
\centering
\begin{tabular}{ccc}
     \includegraphics[width=82mm]{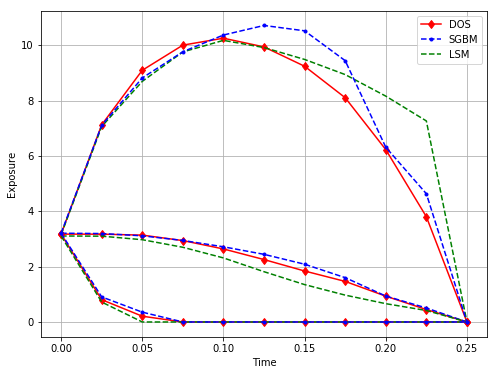}&
          \includegraphics[width=82mm]{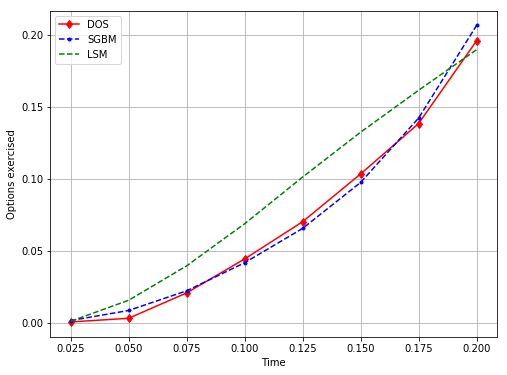}\\
\end{tabular}
\caption{Comparison of the DOS algorithm, the SGBM and the LSM for a Bermudan put option under the Heston model. \textbf{Left:} Expected exposure and potential future exposures at 97.5\%- and 2.5\%-levels. \textbf{Right:} Proportion of options exercised at different exercise dates.}
\label{diff_algo_Heston}
\end{figure}

\section*{Acknowledgments}
This project is part of the ABC-EU-XVA project and has received funding from the European Unions Horizon 2020 research and innovation programme under the Marie Sk\l dowska--Curie grant agreement No 813261. Furthermore, we are greatful for discussions with Adam Andersson, Andrea Fontanari, Lech A. Grzelak, and Shashi Jain regarding the content of this paper.


\begin{thebibliography}{9}

\bibitem{DOS}
S. Becker, P. Cheridito, and A. Jentzen. \emph{Deep Optimal Stopping.} Journal of Machine Learning Research 20.74 (2019): 1-25.

\bibitem{XVA_gregory}
J. Gregory. \emph{The xVA Challenge: Counterparty Credit Risk, Funding, Collateral and Capital.} John Wiley \& Sons, (2015).

\bibitem{XVA_green}
A. Green. \emph{XVA: credit, funding and capital valuation adjustments.} John Wiley \& Sons, (2015).

\bibitem{SGBM_Heston}
C. S. De Graaf, Q. Feng, D. Kandhai, and C.W. Oosterlee. \emph{Efficient computation of exposure profiles for counterparty credit risk.} International Journal of Theoretical and Applied Finance, 17(04), 1450024, (2014).

\bibitem{Yanbin}
Y. Shen, J. A. M. Van Der Weide \& J. H. M. Anderluh \emph{A benchmark approach of counterparty credit exposure of Bermudan option under L\'{e}vy Process: the Monte Carlo-COS Method.} Procedia Computer Science 18 (2013): 1163-1171.

\bibitem{SGBM_P}
Q. Feng, S. Jain, P. Karlsson, D. Kandhai \& C. W. Oosterlee \emph{Efficient computation of exposure profiles on real-world and risk-neutral scenarios for Bermudan swaptions.} Journal of Computational Finance 20(1), 139–172 (2016).

\bibitem{CVA_BER1}
R. Baviera, G. La Bua \& P. Pellicioli. \emph{CVA with wrong-way risk in the presence of early exercise.} Springer, In Innovations in Derivatives Markets (2016): 103-116. 

\bibitem{CVA_BER2}
M. Breton \& M. Oussama. \emph{An efficient method to price counterparty risk.} Groupe d'études et de recherche en analyse des décisions, (2014).

\bibitem{Forsyth}
P. A. Forsyth and K. R. Vetzal. \emph{Quadratic convergence for valuing American options using a penalty method.} SIAM Journal on Scientific Computing, 23(6), (2001): 2095-2122.

\bibitem{Reisinger}
C. Reisinger and J. H. Witte. \emph{On the use of policy iteration as an easy way of pricing American options.} SIAM Journal on Financial Mathematics, 3(1), (2012): 459-478.

\bibitem{Vasquez}
C. Vázquez. \emph{An upwind numerical approach for an American and European option pricing model.} Applied Mathematics and Computation, 97(2-3), (1998): 273-286.

\bibitem{Hout}
T. Haentjens and , K. J. in’t Hout. \emph{ADI schemes for pricing American options under the Heston model.} Applied Mathematical Finance, 22(3), (2015): 207-237.

\bibitem{Hout_book}
K. in't Hout. \emph{Numerical Partial Differential Equations in Finance Explained; An Introduction to Computational Finance.} Palgrave Macmillan UK, (2017).

\bibitem{Fang}
F. Fang and C. W. Oosterlee. \emph{Pricing early-exercise and discrete barrier options by Fourier-cosine series expansions.} Numerische Mathematik, 114(1), 27, (2009).

\bibitem{Zy}
O. Zhylyevskyy. \emph{A fast Fourier transform technique for pricing American options under stochastic volatility. Review of Derivatives Research.} 13(1), (2010): 1-24.

\bibitem{Fang2}
F. Fang and C. W. Oosterlee. \emph{Fourier-based valuation method for Bermudan and barrier options under Heston's model.} SIAM Journal on Financial Mathematics, 2(1), (2011): 439-463.

\bibitem{tree_broadie}
M. Broadie and J. Detemple. \emph{American option valuation: new bounds, approximations, and a comparison of existing methods.} The Review of Financial Studies, 9(4), (1996): 1211-1250.

\bibitem{Rubinstein}
M. Rubinstein. \emph{Edgeworth binomial trees.} Journal of Derivatives, 5, (1998): 20-27.

\bibitem{Jack}
J. C. Jackwerth. \emph{Generalized binomial trees. Journal of Derivatives.}, 5(2), (1996): 7-17.

\bibitem{COD}
R. Bellman. \emph{Dynamic programming.} Science, 153(3731), (1966): 34-37.

\bibitem{Binomial_price}
L. Andersen and M. Broadie. \emph{Primal-dual simulation algorithm for pricing multidimensional American options.} Management Science 50(9), (2004): 1222–1234.

\bibitem{LSM}
F. A. Longstaff, and E. S. Schwartz. \emph{Valuing American options by simulation: a simple least-squares approach.} The review of financial studies 14.1 (2001): 113-147.

\bibitem{BG}
M. Broadie and P. Glasserman. \emph{A Stochastic Mesh Method for Pricing High-Dimensional American Options.} Papers 98-04, Columbia - Graduate School of Business, (1997).

\bibitem{maxcall}
M. Broadie and M. Cao. \emph{Improved lower and upper bound algorithms for pricing American options by simulation.} Quantitative Finance, 8 (2008): 845–861.

\bibitem{SGBM}
S. Jain, and C. W. Oosterlee. \emph{The stochastic grid bundling method: Efficient pricing of Bermudan options and their Greeks.} Applied Mathematics and Computation, 269, (2015): 412-431.

\bibitem{Kohler}
M. Kohler, A. Krzyżak and N. Todorovic. \emph{Pricing of High‐Dimensional American Options by Neural Networks.} Mathematical Finance: An International Journal of Mathematics, Statistics and Financial Economics, 20(3), ISO 690, (2010): 383-410.

\bibitem{Cheridito}
S. Becker, P. Cheridito, and A. Jentzen. \emph{Pricing and hedging American-style options with deep learning.} arXiv preprint arXiv:1912.11060.
ISO 690, (2019).

\bibitem{Laypere}
B. Lapeyre and J. Lelong. \emph{Neural network regression for Bermudan option pricing} ArXiv Preprint, (2019).

\bibitem{SGBMNN}
V. Lokeshwar, V. Bhardawaj and S. Jain, S. \emph{Neural network for pricing and universal static hedging of contingent claims.} Available at SSRN 3491209.
ISO 690, (2019).

\bibitem{KVA}
A. D. Green, C. Kenyon and C. Dennis \emph{KVA: Capital valuation adjustment.} Risk, December, (2014).

\bibitem{PsandQs}
S. Jain, P. Karlsson and D. Kandhai. \emph{KVA, Mind Your P's and Q's!. Wilmott}, 2019(102), ISO 690, (2019): p. 60-73.

\bibitem{oksendal}
B. Øksendal. \emph{Stochastic differential equations} Springer, Berlin, Heidelberg, (2003).

\bibitem{activation_functions}
C. Nwankpa, W. Ijomah, A. Gachagan, and S. Marshall \emph{Activation functions: Comparison of trends in practice and research for deep learning.} arXiv preprint arXiv:1811.03378, (2018).

\bibitem{Adam}
D. P. Kingma, and J. Ba. \emph{Adam: A method for stochastic optimization.} arXiv preprint arXiv:1412.6980. ISO 690. Comment: Published as a conference paper at the 3rd International Conference for Learning Representations, San Diego, (2015).	

\bibitem{regression_book}
L. Györfi, M. Kohler, A. Krzyzak, and H. Walk. \emph{A distribution-free theory of nonparametric regression.} Springer Science \& Business Media, (2006).

\bibitem{white}
H. White. \emph{Asymptotic theory for econometricians.} Academic press,  (2014).

\bibitem{MonteCarloFinance}
P. Glasserman. \emph{Monte-Carlo Methods in Financial Engineering.} Vol. 53. Springer Science \& Business Media, (2013).

\bibitem{ExRegBroadie}
M. Broadie and J. Detemple. \emph{The valuation of American options on multiple assets.} Mathematical Finance 7.3, (1997): 241-286.

\bibitem{FEM_ex_reg}
I. Arregui, B. Salvador, D. Ševčovič, and C. Vázquez. \emph{PDE models for American options with counterparty risk and two stochastic factors: Mathematical analysis and numerical solution.} Computers \& Mathematics with Applications, (2019).

\bibitem{Heston}
S. Heston. \emph{A closed-form solution for options with stochastic volatility with applications to bond and
currency options}, Rev. Financ. Stud., 6 (1993): 327–343.

\bibitem{Pieterbarg}
L. B. Andersen and V. V. Piterbarg. \emph{Moment explosions in stochastic volatility models.} Finance and Stochastics, 11(1), (2007): 29-50.

\bibitem{CIR}
J. C. Cox, J. E. Ingersoll, and S. A. Ross. \emph{A theory of the term structure of interest rates.}, Econometrica, 53 (1985): 385–407.

\bibitem{Feller}
W. Feller. \emph{Two singular diffusion problems.} Annals of mathematics, (1951): 173-182.

\bibitem{QE_Heston}
L. B. Andersen. \emph{Efficient simulation of the Heston stochastic volatility model.} Journal of Computational Finance, (2007).

\bibitem{Fourier_Ruijter}
 M. J. Ruijter, and C. W. Oosterlee. \emph{Two-dimensional Fourier cosine series expansion method for pricing financial options.} SIAM Journal on Scientific Computing, 34(5), B642-B671.
ISO 690, (2014).

\bibitem{RegNowLater}
P. Glasserman and B. Yu. \emph{Simulation for American options: Regression now or regression later?} In Monte-Carlo and Quasi-Monte-Carlo Methods  Springer, Berlin, Heidelberg, (2002): 213-226.

\bibitem{LSM_exposure_profiles}
 K. H. F. Kan, G. Frank, V. Mozgin, and M. Reesor. \emph{Optimized Least-squares Monte-Carlo (OLSM) for Measuring Counterparty Credit Exposure of American-style Options.} Mathematics-in-Industry Case Studies Journal, Volume 2, (2010): 64-85.

\end{thebibliography}
\end{document}